\documentclass[letterpaper]{article} 
\usepackage{aaai2027}  
\usepackage[hyphens]{url}  
\usepackage{graphicx} 
\usepackage{natbib}  
\usepackage{caption} 
\usepackage{algorithm}
\usepackage{tikz}
\usetikzlibrary{positioning, arrows.meta,calc}
\usepackage{times}
\usepackage{soul}
\usepackage{url}
\usepackage[utf8]{inputenc}
\usepackage{graphicx}
\usepackage{amsmath}
\usepackage{amssymb}
\usepackage{amsfonts}
\usepackage{mathtools}
\usepackage{stmaryrd}
\usepackage{amsthm}
\makeatletter
\AddToHook{env/algorithmic/before}{%
  \def\@currentcounter{ALG@line}%
}
\makeatother
\usepackage{cleveref}
\crefalias{ALG@line}{line}
\usepackage{thm-restate}
\theoremstyle{plain}

\newtheorem*{thm*}{Theorem}
\newtheorem{lem}{Lemma}

\newtheorem{obsv}{Observation}
\theoremstyle{definition}

\newtheorem{red}{Reduction rule}
\Crefname{red}{Reduction rule}{Reduction rules}
\crefname{appendix}{appendix}{appendices}
\Crefname{appendix}{Appendix}{Appendices}
\usepackage{booktabs}
\usepackage{makecell}
\usepackage[switch]{lineno}
\usepackage{xcolor}
\usepackage{algpseudocode}
\usepackage{algpseudocodex}
\usepackage{complexity}
\usepackage{adjustbox}
\usepackage{siunitx}
\usepackage{paralist}
\RequirePackage{xspace}
\newcommand{\n}[1]{{#1}\xspace}
\newcommand{\m}[1]{\ensuremath{#1}\xspace} 
\newcommand{\mapf}{\n{\textsc{MAPF}}}
\newcommand{\algname}{\n{IU-PIBT}}
\newcommand{\searchname}{\n{IU-LaCAM}}
\newcommand{\procname}{\textsf{\algname}}

\newcommand{\lacam}{\n{LaCAM}}
\newcommand{\ilp}{\n{\textsc{ILP}}}
\newcommand{\iumapf}{\n{\textsc{IUMAPF}}}
\newcommand{\riumapf}{\n{\textsc{$r$IUMAPF}}}
\newcommand{\giumapf}{\n{\textsc{GMAPF}}}
\newcommand{\suf}[1]{^{\text{#1}}}

\newcommand{\conf}[1]{\mathcal{Q}_{#1}}
\newcommand{\qto}[1]{\mathcal{Q}^{\text{to}}[#1]}
\newcommand{\qfrom}[1]{\mathcal{Q}^{\text{from}}[#1]}
\newcommand{\goal}[1]{g(#1)}

\DeclareMathOperator*{\dist}{\textsf{dist}}
\DeclareMathOperator*{\nextv}{\textsf{next}}
\newcommand{\nextvr}[1]{\textsf{next}^{#1}}
\DeclareMathOperator*{\argmin}{arg\min}
\newcommand{\valid}{\texttt{VALID}}
\newcommand{\invalid}{\texttt{INVALID}}
\newcommand{\true}{\texttt{TRUE}}
\newcommand{\false}{\texttt{false}}
\newcommand{\nei}[1]{N(#1)}

\newcommand{\neicl}[1]{N[#1]}
\newcommand{\neiclr}[2]{N_{#1}[#2]}
\newcommand{\degv}[1]{\textsf{deg}(#1)}
\newcommand{\indG}[1]{G[#1]}
\newcommand{\project}{S^{\text{pro}}}
\newcommand{\absorb}{S^{\text{abs}}}

\newcommand{\funcname}[1]{\m{\mathsf{#1}}}

\newcommand{\rev}[1]{#1}

\algrenewcommand\algorithmicindent{1em}
\newcommand{\IfSingle}[2]{\State \textbf{if}~#1~\textbf{then}~#2}
\newcommand{\Input}[1]{\item[\textbf{input}:~#1]}

\newcommand{\Params}[1]{\item[\textbf{params}:~#1]}

\newcommand{\Continue}{\n{\textbf{continue}}}

\newcommand{\gl}[1]{\textcolor[gray]{0.6}{#1}}

\algdef{SE}[WHILE]{WhileGl}{EndWhileGl}[1]{%
  \algpx@startIndent\algpx@startCodeCommand\gl{\algorithmicwhile\ #1\ \algorithmicdo}%
}{%
  \algpx@endIndent\algpx@startEndBlockCommand\algorithmicend\ \algorithmicwhile%
}
\pretocmd{\WhileGl}{\algpx@endCodeCommand}{}{}
\pretocmd{\EndWhileGl}{\algpx@endCodeCommand}{}{}
\ifbool{algpx@noEnd}{
  \algtext*{EndWhileGl}
  \apptocmd{\EndWhileGl}{\algpx@endIndent}{}{}
  \pretocmd{\EndWhileGl}{\algpx@endCodeCommand[1]}{}{}
}{}

\newcommand{\pop}{\funcname{pop}}
\newcommand{\push}{\funcname{push}}

\newcommand{\init}{\suf{init}}
\newcommand{\new}{\suf{new}}

\newcommand{\Q}{\m{\mathcal{Q}}}
\newcommand{\N}{\m{\mathcal{N}}}

\newcommand{\open}{\m{\mathit{Open}}}
\newcommand{\explored}{\m{\mathit{Explored}}}
\newcommand{\tree}{\m{\mathit{constraints}}}
\newcommand{\config}{\m{\mathit{config}}}

\usepackage{newfloat}
\usepackage{listings}
\DeclareCaptionStyle{ruled}{labelfont=normalfont,labelsep=colon,strut=off} 
\floatstyle{ruled}
\newfloat{listing}{tb}{lst}{}
\floatname{listing}{Listing}

\usepackage{booktabs}

\title{Distance-Constrained Unlabeled Multi-Agent Pathfinding}

\author{
    Takahiro Suzuki\equalcontrib\corresponding\textsuperscript{\rm 1,2},
    Yuma Tamura\equalcontrib\corresponding\textsuperscript{\rm 1},
    Keisuke Okumura\equalcontrib\corresponding\textsuperscript{\rm 2}
}
\affiliations{
    \textsuperscript{\rm 1}Graduate School of Information Sciences, Tohoku University, Japan \\
    \textsuperscript{\rm 2}National Institute of Advanced Industrial Science and Technology (AIST), Japan\\
}

\nocopyright
\begin{document}

\maketitle

\begin{abstract}
    We study a graph pathfinding problem \textsc{Distance-$r$ Independent Unlabeled Multi-Agent Pathfinding}, finding a set of collision-free paths between two sets where agents must stay at pairwise distance at least $r+1$ at all times.
    This additional constraint, which generalizes collision for classical \mapf, 
    makes feasibility (i.e., whether a solution exists) \PSPACE-complete, in contrast to standard (\emph{unlabeled}) \mapf, where it can be decided in polynomial time.
    We address the challenge via two complementary approaches: \emph{(i)}~a reduction-based optimal algorithm with a feasibility-preserving compression, and \emph{(ii)}~a configuration generator-based search.
    Despite the hardness, empirical results show that our algorithm can handle hundreds of agents in a practical timeframe.
\end{abstract}

\section{Introduction}
Research on \emph{Multi-Agent Pathfinding (\mapf)} has attracted considerable attention in the last decade, driven by the growing demand for modern multi-robot systems.
\mapf aims to find a set of paths, or a \emph{plan} from initial points to the targets for given agents, without collisions at every time.
In response to diverse requirements and constraints arising in practical settings, many variants of \mapf have been proposed~\cite{mapf:stern19}; one of them is the \emph{unlabeled (anonymous)} version of \mapf, in which the agents are identical and a \emph{collision} is defined as multiple agents simultaneously occupying the same vertex.
This problem naturally arises in the scenario of homogeneous robots such as warehouse transport~\cite{mapf:MaK16}, where we need to solve an integrated joint problem that optimizes goal assignment and path planning.

Although \mapf has been widely studied, it often simplifies a collision by ignoring geometric interference, which hinders its direct application.
For example, we need to address the agents that occupy more than one grid cell for real-world deployment~\cite{lamapf:Lehoux24}, safety margins (e.g., drone downwash~\cite{drone:honig17}), and delay tolerance for robust execution~\cite{mapf:Atzmon20}.
These issues motivate a formulation that goes beyond the standard one.

These collision managements can be naturally abstracted to a single distance constraint: any pair of agents comes within a distance $r$ on a graph.
Hence, a feasible placement of agents corresponds to a \emph{distance-$r$ independent set}, which has been actively studied in the context of graph theory and graph algorithms.
We introduce an extension of \emph{unlabeled} \mapf; given two vertex subsets, \textsc{Distance-$r$ Independent Unlabeled Multi-Agent Pathfinding} (\riumapf in short) aims to find a plan between them while keeping the placements distance-$r$ independent at every time.

Despite the case $r{=}0$ has been studied as \emph{Unlabeled} \mapf~\cite{mapf:Yu12} and a polynomial-time optimal algorithm is known for $r{=}0$, this generalization remains largely under-explored.
No algorithm applies to general $r\in\mathbb{N}$, and it is \PSPACE-complete to determine feasibility when $r{=}1$~\cite{reconf:Kristan25}.
This implies not only that the problem is unlikely to be solved in polynomial time, but also that some instances have plans with super-polynomial makespan, i.e., feasibility can depend on extremely long-horizon coordination.

While independent set-based collision modeling provides a natural extension of \mapf towards real-world applications, the resulting increase in computational complexity poses a significant challenge for solving \riumapf.
To this end, following the common practice in \mapf, we propose two complementary approaches to tackle \riumapf, which constitute our contributions.
\emph{(i)}~First, to preserve solution quality guarantee, we present a reduction to \textsc{Integer Linear Programming} (\ilp) that yields a makespan-optimal plan.
This is followed by a feasibility-preserving instance compression that accelerates solution discovery.
Still, this method is effective only when the number of agents is small; there is a scalability barrier towards large instances.
\emph{(ii)}~Therefore, we next present a configuration generator \algname and search algorithm \searchname, which is an analogue to PIBT~\cite{mapf:PIBT} and \lacam~\cite{MAPF:lacam23} for classical labeled \mapf.
\algname rapidly computes the next configuration, but its rule-based nature tends to get stuck.
\searchname resolves this issue by adding a search following \lacam and a livelock detection technique to \algname, thereby quickly solving large-scale instances.

Our reduction formulation provides an approach when the plan quality is required, whereas the configuration generator-based approach offers a method for handling a large swarm of agents.
Together, we establish complementary schemes for \riumapf and enable systematic study of it beyond classical \emph{Unlabeled} \mapf.
In what follows, we present a formulation, related work, algorithms, and evaluations for both approaches in order.
All the proofs are available in the appendix.

\section{Preliminaries}\label{sec: preliminaries}
\paragraph{Notation.}
Let $G = (V, E)$ be a simple, undirected, and finite connected graph.
For a pair of vertices $u,v\in V$, we define $\dist(u,v)$ as the length (i.e., number of edges) of the shortest path between $u$ and $v$.
For a vertex $v \in V$, the \emph{closed neighborhood} is $\neicl{v} = \{w \in V \mid vw \in E\}\cup \{v\}$, and its \emph{degree} is $\degv{v} = |\neicl{v}|-1$. 
We extend the notation; let $\neiclr{r}{v}$ be a set of vertices $u$ such that $\dist(u,v)\le r$ for an integer $r\ge 0$.
We define $\Delta \coloneqq \max_{v \in V} \degv{v}$ and $\Delta_r\coloneqq \max_{v\in V}|\neiclr{r}{v}|$. 
For a vertex set $V' \subseteq V$, the neighborhood is $\neicl{V'} = \bigcup_{v \in V'} \neicl{v}$.
We define $\nextvr{r}(s,t)=\argmin_{v\in \neicl{\nextvr{r-1}(s,t)}}\dist(v,t)$, and $\nextvr{0}(s,t)=s$.
We assume that $\argmin$ deterministically returns exactly one vertex.
We use a notation for a set of continuous integers: let $[i,j]$ be a set $\{i,i+1,\dots, j-1,j\}$ for $i\le j$.

\paragraph{Problem Definition.}
Let $G$ be a graph and $A = [1,n]$ be a set of $n$ agents. 
A \emph{configuration} $\mathcal{Q} = (q_1, \dots, q_n) \in V^n$ is an assignment of each agent to a vertex. 
Although the configuration $\mathcal{Q}$ is defined as a list of vertices, we may also refer to it as a (multi-)set of vertices for convenience.
A configuration $\mathcal{Q}$ is \emph{distance-$r$ independent} if there is no pair of agents $i,j\in A$ such that $\dist(\mathcal{Q}[i],\mathcal{Q}[j])\le r$.

Given a set of initial and target vertices $S,T \subset V$ with $|S|=|T|=n$ and a \emph{radius} $r\in \mathbb{N}$, \riumapf aims to find a finite sequence of configurations $\Pi = [\conf{0}, \conf{1}, \ldots, \conf{t^{\ast}}]$, called a \emph{plan}, where sets $\conf{0}$ and $\conf{t^\ast}$ equals to $S$ and $T$, respectively.
Throughout the plan, each agent can move to one of the closed neighbors of the current vertex, provided that every configuration remains distance-$r$ independent.
Formally, a plan $\Pi$ from $S$ to $T$ in \riumapf holds:
\begin{enumerate}
   \item $\forall k \in [1,t^{\ast}], i \in A: \mathcal{Q}_{k}[i] \in \neicl{\mathcal{Q}_{k-1}[i]}$ (\emph{reachability}). 
   \item $\forall k \in [0,t^{\ast}]: \conf{k}$ is distance-$r$ independent.
\end{enumerate}
Note that vertex- and swap-conflicts commonly used in \mapf~\cite{mapf:stern19} can be omitted due to the distance-$r$ independence of configurations and the anonymous setting.

\section{Related work}\label{sec:relatedwork}
\paragraph{\mapf Extensions.} 
Numerous \mapf variants move beyond the classical abstraction to capture practical constraints such as extended-body robots and uncertainty.
Our work offers another approach to accounting for such extensions.

One line of research extends \mapf beyond point agents.
Large-agent \mapf has been studied under various collision models, including continuous Euclidean distance constraints~\cite{lamapf:agafonov25}, explicit geometric shapes~\cite{lamapf:Li19}, and some algorithms are known for these variants~\cite{drone:honig17,lamapf:Li19}.
While these studies are closely related to our setting in that they generalize collision handling, our formulation represents feasible configurations as distance-$r$ independent sets.
This perspective enables the use of graph-theoretic and combinatorial techniques, which we exploit to derive the tractable algorithm presented in \Cref{sec:FPT}.

In addition, we focus on the unlabeled setting, where generalized collision handling has been largely underexplored.
Since unlabeled \mapf often admits more tractable formulations than its labeled counterpart, as reflected in the complexity of makespan optimization~\cite{mapf:Yu12,mapf:Yu13hard} and in the design of efficient suboptimal algorithms~\cite{mapf:tswap}, developing a planner tailored to the unlabeled setting enables a faster and more scalable approach than adapting labeled algorithms with target assignment.

Another line of work relaxes perfect synchrony, considering robust execution under delays~\cite{mapf:Atzmon20} and time uncertainty~\cite{safemapf:Shahar21,mapf:otimapp22}; a plan of \riumapf is one for \emph{unlabeled} \mapf that allows delays of at most $r$ steps.

\paragraph{Reduction-based Approach} translates \mapf into well-known problem and leverage off-the-shelf solvers, e.g., \textsc{SAT}-based approach for \mapf~\cite{MAPF:surynek16}, and flow-based optimal one for \emph{unlabeled} \mapf~\cite{mapf:Yu12}. 
One of the advantages is that they can readily provide guarantees, often used when the plan quality is required. 

\paragraph{Configuration Generator-based Approach} is one of the successful schemes for large-scale, discrete coordination problems~\cite{mapf:PIBT}, as it rapidly generates the successor configuration from the current one and iterates through execution. 
These algorithms find plans in large instances by localizing the search.
Moreover, \citeauthor{MAPF:lacam23} (\citeyear{MAPF:lacam23}) develops a search algorithm, \lacam, that leverages a configuration generator, addressing long-horizon tasks that are difficult to handle solely with a configuration generator.
Prior works~\cite{mapf:Fu25,cumapf:suzuki25} show the advantages of this approach for \mapf variants beyond classical formulation (including the anonymity of agents), and these approaches partially inspire our work.

\section{Optimal algorithm with \ilp}\label{sec:ILP}
As described in \Cref{sec:relatedwork}, several works on \mapf take a reduction-based approach
to guarantee the solution quality.
In this study, we also adopt a reduction-based approach for {\riumapf} to perform planning with guaranteed solution quality: we choose \ilp as the target of the reduction to handle the distance-$r$ independence constraint flexibly.

First, we define the \emph{Bounded} \riumapf: given a graph $G=(V,E)$, vertex subsets $S, T \subseteq V$, and \emph{time bound} $\tau\in \mathbb{N}$, it determines whether there is a plan with length at most $\tau$.
Let $x_{v,t}$ be a variable that indicates whether an agent is at $v\in V$ at step $t$, that is, $x_{v,t}=1$ when $v\in \conf{t}$, and $x_{v,t}=0$ otherwise.
Also, $f_{u,v,t}$ represents whether an agent at vertex $u$ at step $t$ moves to $v\in \neicl{u}$ at step $t+1$ along $uv\in E$.
Then, the problem can be formulated as follows.
\begin{enumerate}
    \item $\forall (v,t) \in V\times [0,\tau]$, $x_{v,t} \in [0,1]$\label{ilpenum:defx}.
    \item $\forall (u,v,t) \in V^2\times [0,\tau]$, $f_{u,v,t} \in [0,1]$\label{ilpenum:deff}.
    \item The initial (or final) configuration $\conf{0}$ (resp. $\conf{\tau}$) coincides with $S$ ($T$): for $v\in V$, $x_{v,0}=1$ if $v\in S$ ($x_{v,\tau}=1$ if $v\in T$), and $x_{v,0} = 0$ ($x_{v,\tau}=0$) otherwise.
    \item For every step $t\in [0,\tau-1]$, exactly $x_{u,t}$ agents on $u$ moves to some vertex $v\in \neicl{u}$: $\forall u\in V$, $\sum_{v\in N[u]} f_{u,v,t} = x_{u,t}$.
    \item For every step $t\in [1,\tau]$, exactly $x_{u,t}$ agents on $u$ comes from some vertex $v\in \neicl{u}$: $\forall u\in V$, $\sum_{v\in N[u]} f_{v,u,t-1} = x_{u,t}$.
    \item For every step $t\in [0,\tau]$, $\conf{t}$ is distance-$r$ independent: $\forall u,v \in V$ s.t. $\dist(u,v)\le r,~x_{u,t}+x_{v,t}\le 1$\label{ilpenum:galactic}.
\end{enumerate}
Condition \ref{ilpenum:galactic} guarantees distance-$r$ independence of a configuration: for every pair $u,v\in V$ with distance at most $r$, it holds $x_{u,t}=0$ or $x_{v,t}=0$.
Then we set an objective as a constant, to check whether there is a satisfying assignment to these constraints.
It leads to an optimal algorithm, since the time bound for YES-instances is upper bounded by $\binom{|V(G)|}{n}$: the number of vertex subsets with size $n$.
\section{Compression Algorithm}\label{sec:FPT}
The above \ilp adds $O(\Delta_r |V|)$ constraints per step, even when the number of agents $n$ is extremely small.
However, one might think that if $n$ is small, a given graph contains redundant vertices that can be contracted while preserving feasibility. 
Such a reduction results in a more compact \ilp instance.
In this section, we theoretically demonstrate that such a reduction is possible if $\Delta$ and $n$ are sufficiently small, although at the expense of optimality. 
Specifically, we develop an algorithm called a \emph{kernelization} that bounds the size of the reduced graph by a function of $\Delta$ and $n$, via reduction to the \emph{galactic reconfiguration}~\cite{gal:Bartier23}.
It was recently established to design kernelizations for \textsc{Token Sliding}, which asks whether one can reach a target token placement from an initial one by repeatedly sliding a token to an adjacent vertex while maintaining the placements as an independent set~\cite{reconf:hearn05}.
To adapt the framework to \riumapf, we first define 
notations for the algorithms and the galactic variant.
We then present a kernelization for \textsc{Galactic} $1\iumapf$ (\giumapf in short). 
Finally, we generalize this method to the distance-$r$ variant.

\paragraph{Kernelization and {\FPT} algorithm.}
A \emph{parameterized problem} of a decision problem $L$ is a set of instances $(x,\kappa)$, where $x$ is an instance of $L$ and $\kappa\in \mathbb{N}$ is called the \emph{parameter} of the problem.
An algorithm for $\mathcal{P}$ is \emph{fixed-parameter tractable} (\FPT) if the algorithm solves $\mathcal{P}$ in $h(k)\cdot \mathrm{poly}(|x|)$ time for every instance $(x,\kappa)$, where $h$ is some computable function and $|x|$ is the size of $x$.
Moreover, $\mathcal{P}$ admits a \emph{kernelization} if there is a polynomial-time algorithm that takes an instance $(x,\kappa)$ and outputs another instance $(x',\kappa')$ (called a \emph{kernel}) such that
\begin{inparaenum}
    \item[\emph{(i)}] $x$ is a YES-instance of $\mathcal{P}$ if and only if $x'$ is a YES-instance of $\mathcal{P}$, 
    \item[\emph{(ii)}] $|x'|$ is bounded from above by $f(\kappa)$, where $f$ is a computable function, and 
    \item[\emph{(iii)}] $\kappa \ge \kappa'$.
\end{inparaenum}
Kernelization immediately yields an {\FPT} algorithm via exhaustive search.
See \cite{para:Cygan15} for more details.

\paragraph{Notation for \giumapf.}
For a vertex subset $S$ of a graph $G$, let $\indG{S}$ denote the subgraph induced by $S$. 
A graph $G$ is \emph{connected} if there is a path between any two vertices of $G$. 
A \emph{component} of $G$ is a maximal connected subgraph.
A \emph{galactic graph} $G=(P\cup B, E)$ is a graph whose vertex set is the disjoint union of two sets $P$ (called \emph{planets}) and $B$ (called \emph{black holes}).
Recall that, in \riumapf, agents are required to keep each configuration independent; thus, no vertex can be occupied by more than one agent at any time.
We relax this constraint so that only agents on the vertices in $P$ are required to be independent, while each vertex in $B$ may hold at most $n$ agents.
That is, a black hole $b\in B$ behaves as a vertex that absorbs agents located at adjacent vertices.
We say that a configuration $\mathcal{Q}$ is \emph{galactic distance-$r$ independent} if there is no pair $i,j\in A$ such that $\mathcal{Q}[i],\mathcal{Q}[j]\in P$ and $\dist_P(\mathcal{Q}[i],\mathcal{Q}[j])\le r$, where $\dist_{P}(u,v)$ for a pair of vertices $u,v\in P$ is the length of the shortest path in $G[P]$ between $u$ and $v$.
In particular, a galactic distance-$1$ independent configuration is simply called a \emph{galactic independent configuration}.
Given a galactic graph $G$ and two sets of vertices $S, T \subseteq P$ such that $b\notin \neicl{S\cup T}$ for all black hole $b$ and $|S|=|T|=n$, Galactic \riumapf asks whether there is a plan $\Pi = [S = \conf{0}, \conf{1}, \ldots, \conf{t^{\ast}}=T]$ satisfying the following two conditions:
\begin{enumerate}
   \item $\forall k \in [1,t^{\ast}],\forall i \in A: \mathcal{Q}_{k}[i] \in \neicl{\mathcal{Q}_{k-1}[i]}$. 
   \item $\forall k \in [0,t^{\ast}]: \conf{k}$ is galactic distance-$r$ independent.
\end{enumerate}
Note that \giumapf generalizes $1\iumapf$; an instance $(G=(V,E), S, T)$ of $1\iumapf$ can be viewed as an instance $(G'=(P\cup B, E),S,T)$ of \giumapf with $(P,B)=(V, \emptyset)$.
\paragraph{\ilp formulation of \giumapf.}
Observe that {\giumapf} can be reduced to {\ilp} by slightly modifying {$1\iumapf$} formulation as follows:
\begin{inparaenum}
    \item[\emph{Cond.~\ref{ilpenum:defx}:}] $x_{v,t}\in [0,1]$ if $v\in P$, and $x_{v,t}\in [0,n]$ otherwise;
    \item[\emph{Cond.~\ref{ilpenum:deff}:}] $f_{u,v,t}\in [0,n]$ if $u,v\in B$, and otherwise $f_{u,v,t}\in [0,1]$, and;
    \item[\emph{Cond.~\ref{ilpenum:galactic} galactic independence:}] $\forall u,v\in P$ s.t. $\dist_P(u,v)\le 1,~ x_{u,t}+ x_{v,t}\le 1$.
\end{inparaenum}
By the same discussion as \Cref{sec:ILP}, we obtain an optimal algorithm for {\giumapf}; similar for \textsc{Distance-$r$} variant.

\subsection{Kernelization for \giumapf}
In this section, we describe a rule that replaces redundant parts of a given instance $\mathcal{I}=(G=(P\cup B,E), S,T)$ with black holes, and then reduce the size of the resulting instance.
This reduction is useful for the \ilp algorithm, since it decreases the size of the instance and hence reduces the number of variables and constraints.
Given an instance $\mathcal{I}$, we first perform a preprocessing as follows.
For each vertex $v \in P$, we partition the vertices into layers: $L_{i}\coloneq \{v \in P\mid d(v)=i\}$, where
$d(v)\coloneq \min_{s\in S\cup T}\dist_P(s,v)$.
Note that every vertex is contained in exactly one layer.
Under this preprocessing, our kernelization algorithm repeatedly performs the following two reductions: one is a straightforward rule, and the other generates black holes.
For the sake of simplicity, we write $V_p$ for the set of vertices in the planet that are not adjacent to $S$ or $T$, i.e., $P\setminus (\neicl{S\cup T})$, and we write $G_p = G[V_p]$.

\begin{figure}
    \centering
\begin{tikzpicture}[
  scale = 0.9,
  corner/.style={draw, rounded corners=2pt, minimum width=8mm, minimum height=8mm, inner sep=0pt},
  mid/.style={draw, circle, fill=white, minimum size=1.6mm, inner sep=0pt},
  vertex/.style={draw, circle, fill=white, minimum size=1.2mm, inner sep=0pt},
  line/.style={draw, line width=0.6pt}
]

{\scriptsize
\node[corner] (TL) at (0,  1.5) {};
\node (C1) at ($(TL)+(0,-0.25)$) {$C_1$};
\node[corner] (TR) at (1.8,  1.5) {};
\node (C2) at ($(TR)+(0,-0.25)$) {$C_2$};
\node[corner] (B) at (0.9,  0.2) {};
\node (C3) at ($(B)+(0,-0.25)$) {$C_3$};
\coordinate (c) at (0.9, 1);

\draw[line] (TL.east) -- (TR.west);
\draw[line] (TR.south) -- (B.east);
\draw[line] (B.west) -- (TL.south);
\draw[line] (c) -- (TL.south east);
\draw[line] (c) -- (TR.south west);
\draw[line] (c) -- (B.north);

\coordinate (L1) at ($(TL.south)!0.25!(B.west)$);
\coordinate (L2) at ($(TL.south)!0.50!(B.west)$);
\coordinate (L3) at ($(TL.south)!0.75!(B.west)$);

\coordinate (T1) at ($(TL.east)!0.25!(TR.west)$);
\coordinate (T2) at ($(TL.east)!0.50!(TR.west)$);
\coordinate (T3) at ($(TL.east)!0.75!(TR.west)$);

\coordinate (R1) at ($(TR.south)!0.25!(B.east)$);
\coordinate (R2) at ($(TR.south)!0.50!(B.east)$);
\coordinate (R3) at ($(TR.south)!0.75!(B.east)$);

\node[mid] at (L1) {}; \node[mid, fill =red] at (L2) {}; \node[mid] at (L3) {};
\node[mid] at (T1) {}; \node[mid, fill = red] at (T2) {}; \node[mid] at (T3) {};
\node[mid] at (R1) {}; \node[mid, fill = blue] at (R2) {}; \node[mid] at (R3) {};

\node[mid,fill = blue] at (c) {};
\node[mid] at ($0.6*(TL.south east) + 0.4*(c)$) {};
\node[mid] at ($0.6*(TR.south west) + 0.4*(c)$) {};
\node[mid] at ($0.6*(B.north) + 0.4*(c)$) {};

\coordinate (v) at ($0.8*(TL.north west) + 0.2*(TL.east) - (0, 0.1)$);
\node[vertex] at (v) {};
\node at ($(v)+ (0.33, 0.3)$) {$v\in L_7$};

\begin{scope}[xshift = 3cm]
\node[corner] (TL) at (0,  1.5) {};
\node (C1) at ($(TL)+(0,-0.25)$) {$C_1$};
\node[corner] (TR) at (1.8,  1.5) {};
\node (C2) at ($(TR)+(0,-0.25)$) {$C_2$};
\node[corner] (B) at (0.9,  0.2) {};
\node (C3) at ($(B)+(0,-0.25)$) {$C_3$};
\coordinate (c) at (0.9, 1);

\draw[line] (TL.east) -- (TR.west);
\draw[line] (TR.south) -- (B.east);
\draw[line] (B.west) -- (TL.south);
\draw[line] (c) -- (TL.south east);
\draw[line] (c) -- (TR.south west);
\draw[line] (c) -- (B.north);

\coordinate (L1) at ($(TL.south)!0.25!(B.west)$);
\coordinate (L2) at ($(TL.south)!0.50!(B.west)$);
\coordinate (L3) at ($(TL.south)!0.75!(B.west)$);

\coordinate (T1) at ($(TL.east)!0.25!(TR.west)$);
\coordinate (T2) at ($(TL.east)!0.50!(TR.west)$);
\coordinate (T3) at ($(TL.east)!0.75!(TR.west)$);

\coordinate (R1) at ($(TR.south)!0.25!(B.east)$);
\coordinate (R2) at ($(TR.south)!0.50!(B.east)$);
\coordinate (R3) at ($(TR.south)!0.75!(B.east)$);

\node[mid] at (L1) {}; \node[mid, fill] at (L2) {}; \node[mid] at (L3) {};
\node[mid] at (T1) {}; \node[mid, fill] at (T2) {}; \node[mid] at (T3) {};
\node[mid] at (R1) {}; \node[mid, fill = blue] at (R2) {}; \node[mid] at (R3) {};

\node[mid,fill = blue] at (c) {};
\node[mid] at ($0.6*(TL.south east) + 0.4*(c)$) {};
\node[mid] at ($0.6*(TR.south west) + 0.4*(c)$) {};
\node[mid] at ($0.6*(B.north) + 0.4*(c)$) {};

\coordinate (v) at ($0.8*(TL.north west) + 0.2*(TL.east) - (0, 0.1)$);
\coordinate (vy) at (v |- TL.east);
\draw[line] (v) -- (vy) -- (TL.east);
\node[vertex,fill = red] at (v) {};
\node at ($(v)+ (0, 0.25)$) {$v$};
\node[vertex] at (vy) {};
\coordinate (H1) at ($(vy)!0.25!(TL.east)$);
\coordinate (H2) at ($(vy)!0.50!(TL.east)$);
\coordinate (H3) at ($(vy)!0.75!(TL.east)$);

\node[vertex,fill = red] at (H1) {};
\node[vertex] at (H2) {};
\node[vertex] at (H3) {};

\end{scope}

\begin{scope}[xshift = 6cm]
\node[mid, fill = black] (TL) at (0,  1.5) {};
\node[corner] (TR) at (1.8,  1.5) {};
\node (C2) at ($(TR)+(0,-0.25)$) {$C_2$};
\node[corner] (B) at (0.9,  0.2) {};
\node (C3) at ($(B)+(0,-0.25)$) {$C_3$};
\coordinate (c) at (0.9, 1);

\draw[line] (TL.east) -- (TR.west);
\draw[line] (TR.south) -- (B.east);
\draw[line] (B.west) -- (TL.south);
\draw[line] (c) -- (TL.south east);
\draw[line] (c) -- (TR.south west);
\draw[line] (c) -- (B.north);

\coordinate (L1) at ($(TL.south)!0.25!(B.west)$);
\coordinate (L2) at ($(TL.south)!0.50!(B.west)$);
\coordinate (L3) at ($(TL.south)!0.75!(B.west)$);

\coordinate (T1) at ($(TL.east)!0.25!(TR.west)$);
\coordinate (T2) at ($(TL.east)!0.50!(TR.west)$);
\coordinate (T3) at ($(TL.east)!0.75!(TR.west)$);

\coordinate (R1) at ($(TR.south)!0.25!(B.east)$);
\coordinate (R2) at ($(TR.south)!0.50!(B.east)$);
\coordinate (R3) at ($(TR.south)!0.75!(B.east)$);

\node[mid] at (L1) {}; \node[mid, fill =red] at (L2) {}; \node[mid] at (L3) {};
\node[mid] at (T1) {}; \node[mid, fill = red] at (T2) {}; \node[mid] at (T3) {};
\node[mid] at (R1) {}; \node[mid, fill = blue] at (R2) {}; \node[mid] at (R3) {};

\node[mid,fill = blue] at (c) {};
\node[mid] at ($0.6*(TL.south east) + 0.4*(c)$) {};
\node[mid] at ($0.6*(TR.south west) + 0.4*(c)$) {};
\node[mid] at ($0.6*(B.north) + 0.4*(c)$) {};
\node at ($(TL)+ (0, 0.25)$) {$b$};

\end{scope}
}
\end{tikzpicture}
    \caption{An illustration of \Cref{red: BFS} when $r=1$, $n=2$, and $B=\emptyset$. Red (or blue) vertices indicate $S$ (resp. $T$).
    Each square $C_1,C_2,C_3$ represents a component of $G_r$.
    \emph{Left}:~Suppose that a vertex $v \in C_1$ has a large layer, $2n+3=7$.
    \emph{Middle}:~Then there is a path of $7-2=5$ vertices that does not touch $\neicl{S\cup T}$, and if an agent can move to a vertex of $C_1$, it can enter some vertex on this path.
    \emph{Right}:~Then $C_1$ can absorb any agents; behave like a black hole. Thus, we compress $C_1$ to a black hole $b$.
    }
    \label{fig:kernel}
\end{figure}

\begin{red}\label{red: adj}
    If two black holes $u$ and $v$ are adjacent, we contract them into a single black hole.
\end{red}
\begin{red}\label{red: BFS}
    Let $v\in V_p$ be a vertex of a component $C$ of $G_p$. If $v$ is included in the layer $L_k$ such that $k\ge 2n+3$, we replace $C$ with a single black hole $b$, and connect $b$ to every vertex in $\neicl{C}\setminus C$ by adding edges.
\end{red}

Our core idea of \Cref{red: BFS}  is that, if there is a vertex $v\in L_k$ with $k>2$, then there is a path with $k-2$ vertices that is not touched by $\neicl{S\cup T}$.
If $k\ge 2n+3$ (\cref{fig:kernel} left), there is an independent set with $n$ vertices on the path.
Moreover, agents placed on vertices in $C$ can be reassigned to the path without breaking galactic independence (\cref{fig:kernel} center); $C$ can absorb $n$ agents and behave like a black hole.
Thus, we contract $C$ into $b$ (\cref{fig:kernel} right).
\begin{restatable}{lem}{FPTcorrect}\label{lem: FPTcorrect}
    \Cref{red: BFS,red: adj} are both safe, i.e., there is a plan in the graph after reductions if and only if there is a plan in the one before reductions.
\end{restatable}
The reductions terminate in linear time.
Consider an instance $(G,S,T)$ of \giumapf, where \Cref{red: BFS,red: adj} are iteratively applied until no updates are available.
One can observe that every planet $p$ has its layer less than $2n+2$;
otherwise $p$ is contracted into a black hole.
Moreover, the size of $B$ is bounded by $|\neicl{P}\setminus P|$.
This leads to \Cref{thm: kernel}.
\begin{restatable}{thm}{kerneldeg}\label{thm: kernel}
    \giumapf admits a kernel with $\Delta^{O(n)}n$ vertices.
\end{restatable}
Moreover, if we are given an induced grid graph, an induced subgraph of the Cartesian product of two paths,
the upper bound of $|P|$ is much smaller.
Note that induced grid graphs are commonly used in empirical evaluation of the {\mapf} algorithms~\cite{mapf:stern19}.
\begin{restatable}{thm}{kernelgrid}\label{thm: kernelgrid}
    \giumapf admits a kernel with $O(n^3)$ vertices when an induced grid graph is given.
\end{restatable}
\begin{restatable}{cor}{galFPT}\label{cor: FPT}
    {\iumapf} admits an {\FPT} algorithm when parameterized by $\Delta+n$.
\end{restatable}

\paragraph{Extension to \textsc{Distance-$r$} variant.}
Recall that our idea is to find a path of sufficiently long length that can ``absorb'' $n$ agents.
In the distance-$r$ variant, an analogous argument holds for a path of length $(r+1)(n+1)-1$.
Therefore, it suffices to modify Reduction rule~\ref{red: BFS} to find a vertex with the layer $L_k$ with $k \ge (r+1)(n+1)-1+(r+1)=(r+1)(n+2)-1$, ensuring that the path does not touch the distance-$r$ neighbors of $S \cup T$.
From the same discussion on \Cref{thm: kernel}, the problem admits a kernel of size $\Delta^{O(rn)}n$.

\section{Configuration Generator-Based Algorithm}\label{sec:PIBT}
\renewcommand{\algorithmicrequire}{\textbf{Input:}}
\renewcommand{\algorithmicensure}{\textbf{Output:}}
Although the \ilp-based algorithm guarantees solution quality, it has scalability limitations, as we observe in \Cref{sec:ILPexperiment}.
To complement this limitation, we propose a configuration generator-based algorithm \searchname, scalable for large instances.
\searchname consists of the configuration generator \algname that supports distance-$r$ independence and a search scheme \lacam with tuning for the unlabeled setting.
We first present \algname, and then search.

\subsection{Configuration Generator \algname}\label{subsec:alg}
\algname takes the current configuration $\mathcal{Q}^{\text{from}}$, the target $T$, and a bijection $g: A \to T$ (used for a target assignment) as input to compute the configuration $\mathcal{Q}^{\text{to}}$ and a new bijection $g': A \to T$ for the next time step. 
Note that an initial bijection can be obtained by some bipartite matching algorithm (e.g., the Hungarian method) with distance evaluation.
\paragraph{Concept.}
\algname is based on rotation-free PIBT~\cite{mapf:PIBT}, which sequentially determines
$\qto{i}$ of each agent $i\in A$ and forbids rotations of agents.
In this paper, \algname generates $\mathcal{Q}^{\text{to}}$ while maintaining distance-$r$ independence, and avoiding \emph{distance-$r$ rotations}, namely a sequence of agents $(a_1,a_2,\ldots,a_k)$ such that
$\qto{a_{i}}\in\neiclr{r}{\qfrom{a_{i+1}}}$ for $i\in [1,k-1]$, and $\qto{a_k}\in\neiclr{r}{\qfrom{a_{1}}}$.
Further, inspired by the suboptimal \emph{unlabeled} \mapf work~\cite{mapf:tswap}, 
\rev{we insert swaps and rotations of target assignments in \algname.}
\begin{algorithm}[t]
    \caption{\algname, generator for \iumapf}
    \label{alg:PIBT}
    {\small
    \begin{algorithmic}[1]  
        \Require configuration $\mathcal{Q}^{\text{from}}$, goals $T$, \textcolor{blue}{bijection $g$}
        \Ensure configuration $\mathcal{Q}^{\text{to}}$ and new bijection $g'$
        \While{\textcolor{blue}{$\exists$ deadlock}} \textcolor{blue}{rotate targets to resolve it}\label{pibt:deadlock} \Comment{rule 1}
        \EndWhile
        
        \For{$i\in A$}
        \If{$\qto{i}=\bot$}~$\Call{\procname}{i,\textcolor{blue}{\epsilon},g}$ \Comment{$\epsilon$: empty list}     
        \EndIf
        \EndFor
        \State \Return $\mathcal{Q}^{to},g$\label{pibt:top_e}
        \Statex
        \Procedure{\procname}{$i,\textcolor{blue}{S},g$}\label{pibt:proc_s}
        \State Sort $v\in N[\qfrom{i}]$ in ascending order of $\dist(v, \goal{i})$
        \For{$v\in N[\qfrom{i}]$}\label{pibt:loop_candidate}

        \If{$\exists j$ s.t. $\qto{j}\ne \bot \land \qto{j}\in \textcolor{blue}{\neiclr{r}{v}}$}~\textbf{continue};\label{pibt:dist}
        \EndIf
        \If{\textcolor{blue}{$\exists j$ s.t. $\qfrom{j}\in\neiclr{r}{v}\land j\in S$}}~\textbf{continue};\label{pibt:rotation}
        \EndIf
        
        \State $\qto{i}\leftarrow v$, $f\leftarrow \true$\Comment{temporal assignment}\label{pibt:tmp}
        \If{$\exists k\coloneq \Call{swap}{i,v}$}~ swap \textcolor{blue}{$\goal{i}$ and $\goal{k}$}\Comment{rule 2}\label{pibt:swap}
        \EndIf
        \For{\textcolor{blue}{$u\in \neiclr{r}{v}$}}\label{pibt:checknei}
        \If{$\exists j$ s.t. $\qfrom{j}=u\land j\ne i$}\label{pibt:existagent}
        \If{$\qto{j}=\bot$}~$\Call{\procname}{j,\textcolor{blue}{S+[i]},g}$\label{pibt:recursive}
        \EndIf
        \If{$\qto{j}\in \neiclr{r}{v}$}~~$f\leftarrow \false$; \textbf{break};
        \EndIf
        \EndIf
        \EndFor \label{pibt:checknei_end}
        \If{$f=\true$}~\Return $\valid$ \label{pibt:valid}
        \EndIf
        \State\textcolor{blue}{swap back $\goal{i}$ and $\goal{k}$ if $k$ exists;} \label{pibt:loop_e}
        
        \EndFor \label{pibt:loop_candidate_end}
        \State $\qto{i}\leftarrow \qfrom{i}$;
        ~\Return $\invalid$
        \EndProcedure\label{pibt:proc_e}
    \end{algorithmic}
    }
\end{algorithm}

\begin{figure}
    \centering
    \begin{tikzpicture}[scale = 0.5]
  \tikzset{vertex/.style={circle, draw, minimum size=9pt,inner sep=0pt }}
  \tikzset{agent/.style={circle, draw,fill = black!20,minimum size=9pt,inner sep=0pt }}
  \tikzset{>={Stealth[length=0.7mm,width=1mm]}}
  
\tikzset{
  mydashed/.style={dash pattern=on 2pt off 1.5pt}
}

{\scriptsize
\newcommand{\DrawGridGraph}[2]{%
\foreach \j in {1,...,3}{%
  \foreach \i in {1,...,4}{%
    \ifnum\j=2\relax
      \ifnum\i=1\relax
      \else
        \node[vertex] (v\i\j) at ({(\i-1)*#1},{(\j-1)*#2}) {};
      \fi
    \else
      \node[vertex] (v\i\j) at ({(\i-1)*#1},{(\j-1)*#2}) {};
    \fi
  }%
}%

  \draw (v11) -- (v21);
  \draw (v21) -- (v31);
  \draw (v31) -- (v41);
  \draw (v22) -- (v32);
  \draw (v13) -- (v23);
  \draw (v23) -- (v33);
  \draw (v33) -- (v43);
  
  \draw (v31) -- (v32);
  \draw (v41) -- (v42);
  \draw (v22) -- (v23);
  \draw (v32) -- (v33);
  \draw (v42) -- (v43);

  \node[agent] (a1) at (1*#1, 1*#2) {$a$};
  \node[agent] (a2) at (2*#1, 2*#2) {$b$};
  \node[agent] (a3) at (3*#1, 1*#2) {$c$};
  \node[agent] (a4) at (2*#1, 0*#2) {$d$};
  \node[agent] (a5) at (0*#1, 2*#2) {$e$};
  }

\DrawGridGraph{1.2}{1}{v}

\draw[->, black,mydashed] (a1) to (v33);
\draw[->, black,mydashed,bend left = 40] (a2) to (v31);
\path[->, black,mydashed] (a3) edge[loop right, looseness=10, min distance=5mm] (v42);
\draw[->, black,mydashed,bend left = 40] (a4) to (v11);
\path[->, black,mydashed] (a5) edge[loop, out=-50, in=-20, looseness=10, min distance=5mm] (v13);
\draw[->, blue,bend left = 30, very thick] (a1) to (v23);
\draw[-{Stealth[length=3.2mm, width=2.4mm]}, line width=1.5pt]
  (-0.6, -0.3) to[out=240, in=120] (-0.5,-2);
  \draw[-{Stealth[length=3.2mm, width=2.4mm]}, line width=1.5pt]
  (3.7, -1.5) to (5,0.8);

\begin{scope}[yshift = -4cm]
\DrawGridGraph{1.2}{1}{v}
\draw[->, black,mydashed] (a1) to (v33);
\draw[->, black,mydashed,bend left = 40] (a2) to (v31);
\path[->, black,mydashed] (a3) edge[loop right, looseness=10, min distance=5mm] (v42);
\draw[->, black,mydashed,bend left = 40] (a4) to (v11);
\path[->, black,mydashed] (a5) edge[loop, out=-50, in=-20, looseness=10, min distance=5mm] (v13);
\draw[->, blue,bend left = 30] (a1) to (v23);
\draw[->, blue,bend right = 30, very thick] (a2) to (v32);
\end{scope}

\begin{scope}[xshift = 5.5cm]
\DrawGridGraph{1.2}{1}{v}
\draw[->, black,mydashed] (a1) to (v33);
\draw[->, black,mydashed,bend left = 40] (a2) to (v31);
\path[->, black,mydashed] (a3) edge[loop right, looseness=10, min distance=5mm] (v42);
\draw[->, black,mydashed,bend left = 40] (a4) to (v11);
\path[->, black,mydashed] (a5) edge[loop, out=-50, in=-20, looseness=10, min distance=5mm] (v13);
\draw[->, blue,bend left = 30] (a1) to (v23);
\draw[->, blue,bend left = 30, very thick] (a2) to (v43);
\draw[-{Stealth[length=3.2mm, width=2.4mm]}, line width=1.5pt]
  (-0.6, -0.3) to[out=240, in=120] (-0.5,-2);
  \draw[-{Stealth[length=3.2mm, width=2.4mm]}, line width=1.5pt]
  (3.7, -1.5) to (5,0.8);
\end{scope}

\begin{scope}[xshift = 5.5cm, yshift = -4cm]
\DrawGridGraph{1.2}{1}{v}
\draw[->, black,mydashed] (a1) to (v33);
\draw[->, black,mydashed] (a2) to (v42);
\draw[->, black,mydashed] (a3) to (v31);
\draw[->, black,mydashed,bend left = 40] (a4) to (v11);
\path[->, black,mydashed] (a5) edge[loop, out=-50, in=-20, looseness=10, min distance=5mm] (v13);
\draw[->, blue,bend left = 30] (a1) to (v23);
\draw[->, blue,bend left = 30] (a2) to (v43);
\draw[->, blue,bend left = 30, very thick] (a3) to (v41);
\draw[->, blue,bend right = 30, very thick] (a4) to (v21);
\end{scope}

\begin{scope}[xshift = 11cm]
\DrawGridGraph{1.2}{1}{v}
\draw[->, black,mydashed] (a1) to (v33);
\draw[->, black,mydashed] (a2) to (v42);
\draw[->, black,mydashed] (a3) to (v31);
\draw[->, black,mydashed,bend left = 40] (a4) to (v11);
\path[->, black,mydashed] (a5) edge[loop, out=-50, in=-20, looseness=10, min distance=5mm] (v13);
\draw[->, blue,bend left = 30] (a1) to (v23);
\draw[->, red,bend left = 30] (a2) to (v43);
\draw[->, red,bend left = 30] (a3) to (v41);
\draw[->, red,bend right = 30] (a4) to (v21);
\draw[->, blue,bend left = 30, very thick] (a5) to (v23);
\draw[-{Stealth[length=3.2mm, width=2.4mm]}, line width=1.5pt]
  (-0.6, -0.3) to[out=240, in=120] (-0.5,-1.8);
\end{scope}

\begin{scope}[xshift = 11cm, yshift = -4cm]
\DrawGridGraph{1.2}{1}{v}
\draw[->, black,mydashed] (a1) to (v33);
\draw[->, black,mydashed] (a2) to (v42);
\draw[->, black,mydashed] (a3) to (v31);
\draw[->, black,mydashed,bend left = 40] (a4) to (v11);
\path[->, black,mydashed] (a5) edge[loop, out=-50, in=-20, looseness=10, min distance=5mm] (v13);
\draw[->, blue,bend left = 30, very thick] (a1) to (v32);
\draw[->, red,bend left = 30] (a2) to (v43);
\draw[->, red,bend left = 30] (a3) to (v41);
\draw[->, red,bend right = 30] (a4) to (v21);
\path[->, red] (a5) edge[loop left, looseness=10, min distance=5mm] (v13);
\end{scope}

}
{\small
\foreach \x in {1,...,4}{
    \node at ({(\x-1)*1.2},2.7) {\x};
    }

\foreach \x in {1,...,3}{
    \node at (-0.7,{(1-\x)*1+2}) {\x};
    }
}
\end{tikzpicture}
    \caption{An example of the \algname's operation when $r{=}1$. The graph $G$ is a part of a grid, and $v_{\texttt{i-j}}$ denotes the vertex with row $i$ ($\downarrow$) and column $j$ ($\rightarrow$).
    Dashed arrows represent the assignment $g$, red and blue arrows represent the next location of each agent (red if fixed, and blue if temporal).
    \emph{Column 1:} {\algname} calls {\procname} for agent $a$, and temporary decides $\qto{a}=v_{\texttt{1-2}}$. Then it starts to check neighbors, and first calls {\procname} for $b$ with $S=[a]$. Since moving $b$ to $v_{\texttt{2-3}}$ causes a distance-$1$ rotation, this fails.
    \emph{Column 2:} Then {\procname} for $b$ reassigns $\qto{b}=v_{\texttt{1-4}}$, and it swaps targets of $b$ and $c$ since $c=$\textsc{SWAP}$(b,v_{\texttt{1-4}})$. Recursive calls to $c$ and $d$ occur.
    \emph{Column 3:} Calls for $b,c,d$ succeed, and {\procname} next checks whether $e$ can safely move out of $\neicl{\qfrom{a}}$. 
    Agent $e$ cannot move to any vertex and decides to stay at $v_{\texttt{1-1}}$.
    Since $a$ can no longer move to $v_{\texttt{1-2}}$, it next chooses another vertex $v_{\texttt{2-3}}$ as temporal $\qto{a}$.
    This does not break the independence of $\mathcal{Q}^{\text{to}}$; the assignment succeeds.
    This yields the next configuration $\mathcal{Q}^{\text{to}}$ and a new goal assignment $g'$.
    } 
    \label{fig:pibt}
\end{figure}

\paragraph{Algorithm.}
\Cref{alg:PIBT} illustrates \algname, consisting of 
\begin{inparaenum}
    \item[\emph{(i)}] top-level procedure (line \ref{pibt:deadlock}--\ref{pibt:top_e}), which updates the bijection $g$ and the
priorities $p$ assigned to target vertices, and 
    \item[\emph{(ii)}] recursive function (line \ref{pibt:proc_s}--\ref{pibt:proc_e}), which determines the next
location $\qto{i}$ of each agent $i\in A$ while maintaining distance-$r$
independence and avoiding distance-$r$ rotations.
\end{inparaenum}
Blue lines indicate the differences from the original rotation-free PIBT.
A step-by-step example of execution is shown in \Cref{fig:pibt}.

\paragraph{\emph{(i)} Top-level procedure. (line \ref{pibt:deadlock}--\ref{pibt:top_e})}
Given a configuration $\mathcal{Q}^{\text{from}}$, goals $T$, a current bijection $g:A\to T$, {\algname} first detects a \emph{deadlock}.
A deadlock is a cycle of agents such that each agent is blocked by the next agent on its shortest path of length $(r+1)$
toward the assigned target.
When such a cycle is found, {\procname} rotates the targets along the cycle (\cref{fig_rules}a right).
The deadlock can be detected in a depth-first manner: see the Appendix for details.
Then, {\algname} \rev{sequentially} calls {\procname} for every agent $i$ whose next location remains undecided \rev{(i.e., $\qto{i}=\bot$)}.

\paragraph{\emph{(ii)} Procedure \procname. (line \ref{pibt:proc_s}--\ref{pibt:proc_e})}
For an agent $i$, the recursive procedure scans the vertices in $\neicl{\qto{i}}$ in increasing order of their distance to $\goal{i}$.
A candidate $v$ is rejected if there is an agent $j$ such that $\qto{j}\in \neiclr{r}{v}$ or $\qfrom{j}\in \neiclr{r}{v}$ and $j\in S$.
The latter one prevents distance-$r$ rotation, since accepting such a candidate would cause a cyclic dependency along the current recursion.
If $v$ passes these tests, {\procname} temporarily sets $\qto{i}\leftarrow v$ and recursively processes agents
currently located in $\neiclr{r}{v}$. 
If all of them obtain next locations outside $\neiclr{r}{v}$, then $v$ is fixed as $\qto{i}$; otherwise, the choice is cancelled, and it tries the next candidate. 
If no candidate succeeds, $i$ stays at $\qfrom{i}$.

\paragraph{Target swapping. (line \ref{pibt:swap}, \ref{pibt:loop_e})} 
During the recursive procedure, {\algname} may temporarily swap targets.
When an agent $i$ already located at $\goal{i}$ blocks another agent $j$ from approaching its target $\goal{j}$ (\cref{fig_rules}b center), {\procname} swaps $\goal{i}$ and $\goal{j}$ before the recursive calls.
The swap is kept only if the recursive call succeeds (\cref{fig_rules}b left), and is reverted otherwise (right).
Pseudocode for the function $\textsc{SWAP}$ is provided in the Appendix.

\begin{figure}
\begin{minipage}[t]{0.5\columnwidth}
    \begin{tikzpicture}[scale = 0.42]
  \tikzset{vertex/.style={circle, draw, minimum size=8pt,inner sep=0pt }}
  \tikzset{agent/.style={circle, draw,fill = black!20,minimum size=8pt,inner sep=0pt }}
  \tikzset{>={Stealth[length=1mm,width=1mm]}}
  \tikzset{
  thickarrow/.style={
    line width=3pt,
    >={Stealth[length=2mm,width=5mm]}
  }
}
\tikzset{
  mydashed/.style={dash pattern=on 2pt off 1.5pt}
}
\begin{scope}[rotate = 90]
{\scriptsize
  \node[vertex] (v1) at (0,0) {};
  \node[vertex] (v2) at (-0.75,-0.75) {};
  \node[vertex] (v3) at (-1.5,-1.5) {};
  \node[vertex] (g1) at (-2.25,-2.25) {};
  \node[vertex] (v4) at (-0.75,-2.25) {};
  \node[vertex] (v5) at (0,-3) {};
  \node[vertex] (g3) at (0.75,-3.75) {};
  \node[vertex] (v6) at (0.75,-2.25) {};
  \node[vertex] (v7) at (1.5,-1.5) {};
  \node[vertex] (g5) at (2.25,-0.75) {};
  \node[vertex] (v8) at (0.75,-0.75) {};
  \node[vertex] (g7) at (-0.75, 0.75) {};
  \node[agent] (a1) at (0,0) {$i$};
  \node[agent] (a2) at (-1.5,-1.5) {$j$};
  \node[agent] (a3) at (0,-3) {};
  \node[agent] (a4) at (1.5,-1.5) {};
  
  \foreach\x in {1,2,3,4,5,6,7,8}{
  \pgfmathtruncatemacro{\y}{mod(\x, 8)+1};
    \draw (v\x) -- (v\y); 
  }
  \foreach\x in {1,3,5,7}{
  \pgfmathtruncatemacro{\y}{mod(\x+6, 8)};
  \draw[->, black, bend right=60] (v\x) to (g\x);
   \draw[mydashed] (v\x) -- (g\y); 
  } 
\draw[->,thickarrow, bend left=40] (1.8,-3.5) to (1.8,-5.5);
\begin{scope}[yshift = -6cm]
    \node[vertex] (v1) at (0,0) {};
  \node[vertex] (v2) at (-0.75,-0.75) {};
  \node[vertex] (v3) at (-1.5,-1.5) {};
  \node[vertex] (g1) at (-2.25,-2.25) {};
  \node[vertex] (v4) at (-0.75,-2.25) {};
  \node[vertex] (v5) at (0,-3) {};
  \node[vertex] (g3) at (0.75,-3.75) {};
  \node[vertex] (v6) at (0.75,-2.25) {};
  \node[vertex] (v7) at (1.5,-1.5) {};
  \node[vertex] (g5) at (2.25,-0.75) {};
  \node[vertex] (v8) at (0.75,-0.75) {};
  \node[vertex] (g7) at (-0.75, 0.75) {};
  \node[agent] (a1) at (0,0) {$i$};
  \node[agent] (a2) at (-1.5,-1.5) {$j$};
  \node[agent] (a3) at (0,-3) {};
  \node[agent] (a4) at (1.5,-1.5) {};
  
  \foreach\x in {1,2,3,4,5,6,7,8}{
  \pgfmathtruncatemacro{\y}{mod(\x, 8)+1};
    \draw (v\x) -- (v\y); 
  }
  \foreach\x in {1,3,5,7}{
  \pgfmathtruncatemacro{\y}{mod(\x+6, 8)};
  \draw[->, black, bend right=60] (v\x) to (g\y);
   \draw[mydashed] (v\x) -- (g\y); 
  } 
\end{scope}

}
\end{scope}
\node at (4.5,-3.5) {(a)};
\end{tikzpicture}
    \end{minipage}
    \hspace{0.07\columnwidth}
    \begin{minipage}[t]{0.17\columnwidth}
        \begin{tikzpicture}[scale = 0.25]
  \tikzset{vertex/.style={circle, draw, minimum size=8pt,inner sep=0pt }}
  \tikzset{agent/.style={circle, draw,fill = black!20,minimum size=8pt,inner sep=0pt }}
  \tikzset{>={Stealth[length=2mm,width=1mm]}}
  \tikzset{
  thickarrow/.style={
    line width=3pt,
    >={Stealth[length=2mm,width=3mm]}
  }
}
\tikzset{
  mydashed/.style={dash pattern=on 2pt off 1.5pt}
}
\begin{scope}[rotate = 270,yscale = 1.15]

{\scriptsize
  \node[vertex] (v1) at (0,0) {};
  \node[vertex] (v2) at (2,0) {};
  \node[vertex] (v3) at (4,0) {};
  \node[vertex] (v4) at (6,0) {};
  \node[vertex] (v5) at (8,0) {};
  
  \node[agent] (a1) at (0,0) {$i$};
  \node[agent] (a2) at (4,0) {$j$};

  \draw[mydashed] (v4) -- (v5);
  \draw (v1) -- (v2);
  \draw (v2) -- (v3);
  \draw (v3) -- (v4);

  \draw[->, blue, bend left=40] (v1) to (v2);
  \draw[->, black, mydashed, bend right=30] (v1) to (v5);
  \path[->, black, mydashed] (v3) edge[loop above, looseness=15, min distance=8mm] (v3);

\draw[->,thickarrow, bend right = 30] (-0.5,-1) to node[midway,xshift = -8pt, yshift=7pt, inner sep=2pt] {$j:\invalid$} (-0.5,-3.5);
\begin{scope}[yshift = 4cm]
  \node[vertex] (v1) at (0,0) {};
  \node[vertex] (v2) at (2,0) {};
  \node[vertex] (v3) at (4,0) {};
  \node[vertex] (v4) at (6,0) {};
  \node[vertex] (v5) at (8,0) {};
  
  \node[agent] (a1) at (0,0) {$i$};
  \node[agent] (a2) at (4,0) {$j$};

  \draw[mydashed] (v4) -- (v5);
  \draw (v1) -- (v2);
  \draw (v2) -- (v3);
  \draw (v3) -- (v4);

  \draw[->, red, bend left=40] (v1) to (v2);
  \draw[->, red, bend left = 60] (v3) to (v4);
  \draw[->, black, mydashed, bend right=60] (v1) to (v3);
  \draw[->, black, mydashed, bend right=60] (v3) to (v5);
  \end{scope}

\draw[->,thickarrow, bend left = 30] (-0.5,1) to node[midway,xshift = 8pt, yshift=7pt, inner sep=2pt] {\makecell[c]{$j:\valid$}} (-0.5,3.5);
    \begin{scope}[yshift = -4.3cm]
  \node[vertex] (v1) at (0,0) {};
  \node[vertex] (v2) at (2,0) {};
  \node[vertex] (v3) at (4,0) {};
  \node[vertex] (v4) at (6,0) {};
  \node[vertex] (v5) at (8,0) {};
  
  \node[agent] (a1) at (0,0) {$i$};
  \node[agent] (a2) at (4,0) {$j$};

  \draw[mydashed] (v4) -- (v5);
  \draw (v1) -- (v2);
  \draw (v2) -- (v3);
  \draw (v3) -- (v4);

  \path[->, blue] (v1) edge[loop below, looseness=15, min distance=8mm] (v1);
  \path[->, mydashed, black] (v3) edge[loop above, looseness=15, min distance=8mm] (v3);
  \path[->, red] (v3) edge[loop below, looseness=15, min distance=8mm] (v3);
  \draw[->, mydashed, black, bend right=30] (v1) to (v5);
  
  \end{scope}
}

\end{scope}
  \node at (0,-9.5) {(b)};
\end{tikzpicture}
    \end{minipage}
    \caption{Two rules in \procname, when $r=1$. Black arrows represent the assignment of goals to each agent, red and blue arrows represent the next location of each agent (red if fixed, and blue if temporal). \emph{(a)} Before (left) and after (right) performing deadlock resolution. If a deadlock is found, we resolve it by rotating the assignment of targets. \emph{(b)} Target swapping. When an agent $j$ has already reached its goal (center), it swaps $\goal{i}$ and $\goal{j}$
    and calls {\procname} for $j$. When the call for $j$ succeeds (right), we retain the modified assignment; otherwise (left), we revert to the original one.
    }
    \label{fig_rules}
\end{figure}

\paragraph{Theoretical analysis.}\label{sec:proof}
The following properties suggest that \algname is a sound and computationally efficient configuration generator:
\begin{restatable}{lem}{correctconf}\label{lem: correctconf}
    If $\mathcal{Q}^{\text{from}}$ is distance-$r$ independent, then $\mathcal{Q}^{\text{to}}$ is reachable and a distance-$r$ independent configuration.
\end{restatable}
\begin{restatable}{lem}{runtime}\label{lem: runtime}
    \algname outputs $\mathcal{Q}^{\text{to}}$ in $O(\Delta_r\Delta (r+1)n +\alpha)$ time, where $\alpha$ denotes the time for resolution of deadlock.
\end{restatable}

\begin{figure}[t]
\begin{minipage}[t]{0.25\columnwidth}
    \begin{tikzpicture}[scale = 0.5]
  \tikzset{vertex/.style={circle, draw, minimum size=8pt,inner sep=0pt }}
  \tikzset{agent/.style={circle, draw,fill = black!20,minimum size=8pt,inner sep=0pt }}
  \tikzset{>={Stealth[length=1mm,width=1mm]}}
  \tikzset{
  thickarrow/.style={
    line width=3pt,
    >={Stealth[length=2mm,width=5mm]}
  }
}
\tikzset{
  mydashed/.style={dash pattern=on 2pt off 1.5pt}
}
{\scriptsize
  \node[vertex] (v1) at (0,0) {};
  \node[vertex] (v2) at (-0.75,-0.75) {};
  \node[vertex] (v3) at (-1.5,-1.5) {};
  \node[vertex] (v4) at (-0.75,-2.25) {};
  \node[vertex] (v5) at (0,-3) {};
  \node[vertex] (v6) at (0.75,-2.25) {};
  \node[vertex] (v7) at (1.5,-1.5) {};
  \node[vertex] (v8) at (0.75,-0.75) {};
  \node[agent] (a1) at (0,0) {};
  \node[agent] (a2) at (-1.5,-1.5) {};
  \node[agent] (a3) at (0,-3) {};
  \node[agent] (a4) at (1.5,-1.5) {};
  
  \foreach\x in {1,2,3,4,5,6,7,8}{
  \pgfmathtruncatemacro{\y}{mod(\x, 8)+1};
    \draw (v\x) -- (v\y); 
  }
  \foreach\x in {1,3,5,7}{
  \pgfmathtruncatemacro{\y}{\x+1};
  \draw[->, black, bend right=60] (v\x) to (v\y);
  } 
}
\node at (-2,0) {(a)};
\end{tikzpicture}
    \end{minipage}
    \hspace{0.05\columnwidth}
    \begin{minipage}[t]{0.75\columnwidth}
        \begin{tikzpicture}[scale = 0.5]
  \tikzset{vertex/.style={circle, draw, minimum size=8pt,inner sep=0pt }}
  \tikzset{agent/.style={circle, draw,fill = black!20,minimum size=8pt,inner sep=0pt }}
  \tikzset{>={Stealth[length=1mm,width=1mm]}}
  \tikzset{
  thickarrow/.style={
    line width=3pt,
    >={Stealth[length=2mm,width=5mm]}
  }
}
\tikzset{
  mydashed/.style={dash pattern=on 2pt off 1.5pt}
}
{\scriptsize
  \begin{scope}[rotate = 270]
  \node[vertex] (v1) at (1,0) {};
  \node[vertex] (v2) at (2,0) {};
  \node[vertex] (v3) at (3,0) {};
  \node[vertex] (v4) at (4,0) {};
  \node[vertex] (u1) at (1,-1) {};
  \node[vertex] (u2) at (2,-1) {};
  \node[vertex] (u3) at (3,-1) {};
  \node[vertex] (u4) at (4,-1) {};
  
  \node[agent] (a1) at (2,0) {$i$};
  \node[agent] (a2) at (1,-1) {$j$};
  \foreach\x in {1,2,3}{
  \pgfmathtruncatemacro{\y}{\x+1};
    \draw (v\x) -- (v\y); 
    \draw (u\x) -- (u\y);
  }
  \foreach\x in {1,2,3,4}{
    \draw (v\x) -- (u\x); 
  }
  \path[->, black] (a1) edge[loop above, looseness=10, min distance=5mm] (v2);
  \draw[->, black, bend right=60] (a2) to (u3);
  \end{scope}
\begin{scope}[xshift = 3.5cm, rotate = 270]
    \node[vertex] (v1) at (1,0) {};
  \node[vertex] (v2) at (2,0) {};
  \node[vertex] (v3) at (3,0) {};
  \node[vertex] (v4) at (4,0) {};
  \node[vertex] (u1) at (1,-1) {};
  \node[vertex] (u2) at (2,-1) {};
  \node[vertex] (u3) at (3,-1) {};
  \node[vertex] (u4) at (4,-1) {};
  
  \node[agent] (a1) at (3,0) {$i$};
  \node[agent] (a2) at (2,-1) {$j$};
  \foreach\x in {1,2,3}{
  \pgfmathtruncatemacro{\y}{\x+1};
    \draw (v\x) -- (v\y); 
    \draw (u\x) -- (u\y);
  }
  \foreach\x in {1,2,3,4}{
    \draw (v\x) -- (u\x); 
  }
  \draw[->, black, bend right=60] (a1) to (v2);
  \draw[->, black, bend right=60] (a2) to (u3);
\end{scope}

\begin{scope}[xshift = 7cm, rotate = 270]
    \node[vertex] (v1) at (1,0) {};
  \node[vertex] (v2) at (2,0) {};
  \node[vertex] (v3) at (3,0) {};
  \node[vertex] (v4) at (4,0) {};
  \node[vertex] (u1) at (1,-1) {};
  \node[vertex] (u2) at (2,-1) {};
  \node[vertex] (u3) at (3,-1) {};
  \node[vertex] (u4) at (4,-1) {};
  
  \node[agent] (a1) at (4,0) {$i$};
  \node[agent] (a2) at (3,-1) {$j$};
  \foreach\x in {1,2,3}{
  \pgfmathtruncatemacro{\y}{\x+1};
    \draw (v\x) -- (v\y); 
    \draw (u\x) -- (u\y);
  }
  \foreach\x in {1,2,3,4}{
    \draw (v\x) -- (u\x); 
  }
  \draw[->, black, bend right=60] (a1) to (v2);
    \path[->, black] (a2) edge[loop below, looseness=10, min distance=5mm] (u3);
\end{scope}
}

  \draw[
    <->,
    line width=1.2pt,
    >=Stealth,
    bend left=25
  ] (0.5,-1) to (2,-1);
  \draw[
    <->,
    line width=1.2pt,
    >=Stealth,
    bend left=25
  ] (4,-1) to (5.5,-1);
\node at (-2.5,-1) {(b)};
\end{tikzpicture}
    \end{minipage}
    \caption{Adversarial instances for \algname when $r=1$. \emph{(a)}~One needs a rotation, unable to solve with \algname. The resulting configuration $\mathcal{Q}^{\text{to}}$ is equal to $\mathcal{Q}^{\text{from}}$. \emph{(b)}~Livelock. If \rev{$j$ is called first in \procname} $i$ is forced to move down. Otherwise, $j$ is forced to move up. 
    This process recurs infinitely, ensuring that $i$ and $j$ will never reach the goal simultaneously.
    }
    \label{fig_adv}
\end{figure}

\paragraph{Adversarial instances.}
Since \algname enforces rotation-free moves for agents, there are simple local structures where a configuration does not update (\Cref{fig_adv}a) or where a livelock occurs (\Cref{fig_adv}b) when rotation is required.
While {\algname} serves as a good configuration generator for $r=0$ (see \Cref{sec:exptswap}), it easily gets stuck for general $r$ due to livelocks; this motivates us to integrate \algname with a systematic search over configurations.

\subsection{\searchname}\label{sec:lacam}

{\lacam}~\cite{MAPF:lacam23} is a well-known complete \mapf solver that uses a configuration generator with \emph{constraints}, which include \emph{a subset of agents and their next locations} for creating a successor of the current one; by changing the constraints, it can generate different successors of the same configuration. 
\lacam progresses by updating a stack OPEN that stores the search nodes, each of which corresponds to one configuration $\mathcal{Q}$, and maintains search states as tuples $\langle \mathcal{Q}, \llbracket \mathcal{C} \rrbracket \rangle$ during the search, where $\mathcal{C}$ is a list of constraints.
If the generator produces a configuration identical to one already generated, it modifies the constraints to induce the generation of a different configuration. Eventually, it tries all possible constraints for every configuration, thereby generating all successors, yielding completeness.
It was originally developed for \emph{labeled} \mapf; however, we can also adapt it to the \emph{unlabeled} setting by including the target assignment $g$ in the search state. 
This leads to a complete algorithm by regarding each configuration as a vertex set, since a set of constraints can generate all neighboring configurations.
Its pseudocode is in the Appendix.
{
\renewcommand{\S}{\m{\mathcal{S}}}
\begin{algorithm}[t!]
\caption{livelock detection and target reassignment}
\label{algo:livelock}
\begin{algorithmic}[1]
\small
\Input{node $\mathcal{N}:\langle \mathcal{Q},\llbracket\mathcal{C}\rrbracket,g, \mathcal{B}\rangle$}\Comment{when generated known config $\mathcal{Q}$}
\Params{depth of detection $d \in \mathbb{N}_{\geq 0}$}
\State $\mathcal{N}^{\text{ans}} \leftarrow \textsf{parent}(\mathcal{N})$  \Comment{$\mathcal{N}^{\text{ans}} = \langle \mathcal{Q}^{\text{ans}}, [\mathcal{C}^{\text{ans}}], g^{\text{ans}},\mathcal{B}^{\text{ans}} \rangle$}
\While{$\mathcal{N}^{\text{ans}}\ne \bot$}
\If{$\mathcal{Q}=\mathcal{Q}^{\text{ans}} \land g = g^{\text{ans}}$}\label{livelock:visited}
\State $D \leftarrow \{i \mid \mathcal{Q}[i] = \mathcal{Q}^{\text{ans}}[i] \land \mathcal{Q}[i]\ne \goal{i}\}$
\State $\mathcal{B}'\leftarrow [\mathcal{B}^{\text{ans}}[i] \cup \{\goal{i} \mid i\in D\}]$\label{livelock:banned}
\State find an assignment $g'$ s.t. $g'(i)\notin \mathcal{B}'[i]$\label{livelock:rematch}
\If{succeeded} $\mathcal{N}^{\text{ans}} \leftarrow\langle \mathcal{Q},\epsilon,g', \mathcal{B}'\rangle$
\State push $\mathcal{N}^{\text{ans}}$ to stack $\text{OPEN}$; \Return \label{livelock:reset}
\EndIf
\EndIf
\State $\mathcal{N}^{\text{ans}} \leftarrow \textsf{parent}(\mathcal{N}^{\text{ans}})$; $d\leftarrow d-1$
\EndWhile
\end{algorithmic}
\end{algorithm}
}

\paragraph{Problem-specific adaptation.}
While \lacam-style search is applicable as above, a prior work~\cite{mapf:lagat25} shows that the livelock resolution scheme potentially accelerates the search.
Inspired by this observation, we introduce a similar technique into {\searchname} by reconstructing the target assignment.
A concrete strategy is shown in \Cref{algo:livelock}.
Each search node of \searchname additionally stores a list $\mathcal{B}$ of forbidden targets for each agent.
If a generated configuration is identical to one of the most recent $d$ ($=2$ in our experiment) configurations and the assignment $g$ is also identical, then we consider the search to have entered a livelock.
It then identifies agent $i$ that has not reached their targets $\goal{i}$ and has not moved, and adds $\goal{i}$ to the corresponding list $\mathcal{B}[i]$. If a new assignment $g'$ avoiding these targets exists, \searchname reinserts the node with $g'$ and resets constraints (line \ref{livelock:reset}).
Note that all of the modifications in this section do not break the completeness of \lacam, deriving the following:

\begin{restatable}{thm}{lacamcomplete}\label{thm:lacamcomplete}
    \searchname is complete for \riumapf.
\end{restatable}

\section{Empirical Evaluation}
\begin{figure*}[t]
    \centering
    \input{draft/tikz/fig_expmain}
    \caption{Time and makespan across three maps. For larger $r$, we cannot secure a sufficient number of distance-$r$ independent sets
    for extremely dense settings; thus, the upper bound on $n$ differs among $r$, and the success rate never falls below 1.0 in some settings. We also report the 95\% CI of the makespan by filled area, for the case where the success rate exceeds 0.1.
    Note that a denser scenario can have a lower makespan than a sparser one in the unlabeled setting.}
    \label{fig:expmain}
\end{figure*}

We evaluate both \ilp- and generator-based methods.
These are coded in Python, and the experiments were run on a Mini PC with Intel Core i9-13900H \SI{2.6}{\giga\hertz} CPU and \SI{32}{\giga\byte} RAM.
For our evaluation, we use \emph{empty-16-16}, \emph{random-64-64-20}, \emph{lak303d}, and \emph{warehouse-10-20-10-2-2} in \mapf benchmarks~\cite{mapf:stern19}.
Since the scenario in \mapf benchmarks does not apply to \riumapf, we generate instances by sampling two random distance-$r$ independent sets ($S$ and $T$), for each map, distance $r$, and number of agents $n$.

\paragraph{Baselines.}
\emph{Conflict-based search (CBS)}~\cite{mapf:CBS} is a promising baseline for our problem, as it has been applied to several \mapf extensions~\cite{lamapf:Li19,drone:honig17}; we tested an adaptation of CBS by converting an unlabeled instance to the labeled one with the Hungarian target assignment.
We also tested a \emph{greedy} variant of CBS~\cite{GCBS/BarerSSF14}, which is used for a suboptimal baseline.
Furthermore, as a generator-based baseline, we tested a solver that successively applies \algname.

\subsection{Evaluation of Exact \ilp algorithms}\label{sec:ILPexperiment}
\begin{table}[t!]
    \centering
\setlength{\tabcolsep}{2pt}
\renewcommand{\arraystretch}{1}
\begin{adjustbox}{max width=\textwidth}
{\small
\begin{tabular}{crrrrrr!{\hspace{0.5em}}rrr}
\toprule
&&& \multicolumn{2}{c}{\ilp} &\multicolumn{2}{c}{Compression  }&\multicolumn{3}{c}{\searchname}\\
\cmidrule{4-10}
  Map & $n$ & $r$ & rate & time (\SI{}{\second}) & rate & time & rate & time & subopt \\
\midrule
  \makecell[c]{{\scriptsize\textit{\makecell[c]{empty-\\16-16}}}\\ \adjustbox{raise=-5mm}{\includegraphics[width = 0.1\linewidth]{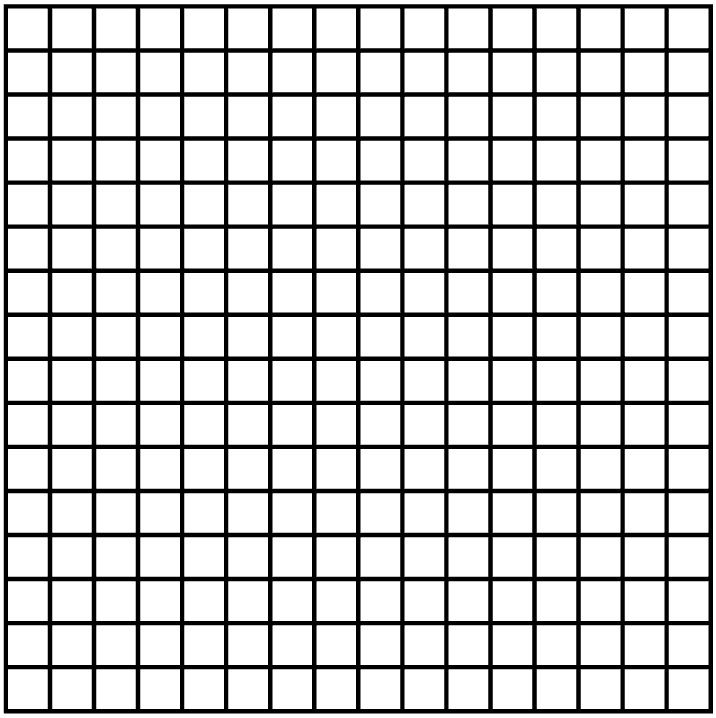}}} 
& \makecell[r]{10\\\\20\\\\30\\}& \makecell[r]{1\\2\\1\\2\\1\\2} &\makecell[c]{\textbf{1.00}\\\textbf{1.00}\\\textbf{1.00}\\\textbf{1.00}\\\textbf{1.00}\\\textbf{1.00}}&\makecell[c]{0.48\\0.80\\0.41\\0.59\\0.39\\0.53}&\makecell[c]{\textbf{1.00}\\\textbf{1.00}\\\textbf{1.00}\\\textbf{1.00}\\\textbf{1.00}\\\textbf{1.00}}&
\makecell[r]{0.33\\0.69\\0.26\\0.48\\0.23\\0.41}&
\makecell[c]{\textbf{1.00}\\\textbf{1.00}\\\textbf{1.00}\\\textbf{1.00}\\\textbf{1.00}\\\textbf{1.00}}&
\makecell[c]{0.030\\0.032\\0.037\\0.046\\0.045\\0.043}&\makecell[c]{1.33\\1.42\\1.83\\1.57\\2.12\\1.95}\\
\midrule
  \makecell[c]{{\scriptsize\textit{\makecell[c]{random-\\64-64-20}}}\\\adjustbox{raise=-5mm}{\includegraphics[width = 0.1\linewidth]{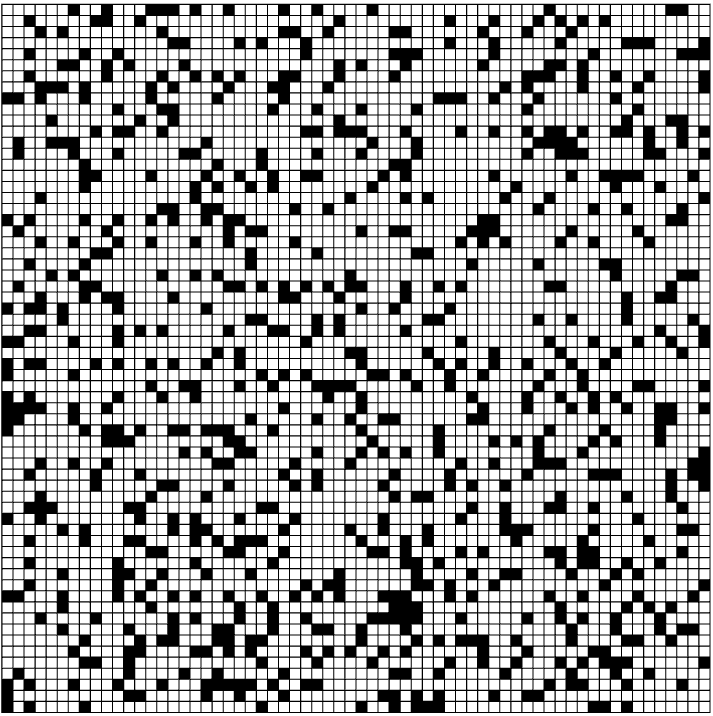}}} 
& \makecell[c]{10\\\\20\\\\30\\}& \makecell[r]{1\\2\\1\\2\\1\\2}&\makecell[c]{0.86\\0.74\\0.90\\0.56\\0.84\\0.76}  &
\makecell[c]{41.6\\44.7\\41.6\\45.8\\40.4\\42.6} &\makecell[c]{\textbf{1.00}\\0.82\\0.92\\0.66\\0.90\\0.82} &
\makecell[r]{0.056\\23.8\\20.0\\34.6\\21.1\\29.6}&
\makecell[c]{\textbf{1.00}\\\textbf{1.00}\\\textbf{1.00}\\\textbf{1.00}\\\textbf{1.00}\\\textbf{1.00}}&\makecell[c]{0.091\\0.088\\0.15\\0.16\\0.21\\0.22}&
\makecell[c]{1.33\\1.30\\1.41\\1.31\\1.61\\1.74}\\
\bottomrule
\end{tabular}
}
\end{adjustbox}
\caption{
Evaluation of the \ilp algorithms. We prepared 50 random instances for each map and $r\in \{1,2\}$, $n\in \{10,20,30\}$, and report the \emph{rate} of instances solved within \SI{60}{\second} and the average running \emph{time} (\si{\second}) over those instances. We also report rate, time, and suboptimality over solved instances (\emph{subopt}) of \searchname.}\label{tab:ILP}
\end{table}
\Cref{tab:ILP} shows both the effectiveness and limitations of \ilp approach, which uses Gurobi as the \ilp solver.
We observe that, on a simple map (\emph{empty-16-16}), \ilp can find an optimal solution within a second, while on a large-scale instance (\emph{random-64-64-20}), it often fails to find a solution, and the number of failed instances increases as $n$ increases.
On large maps, the compression can partially mitigate this issue (see e.g., \emph{random-64-64-20} with $n=10$).
However, this improvement occurs only in settings where $n/ |V|$ is extremely small.
These results indicate that the scalability of the \ilp approach is primarily limited by the growth in the number of variables, and it is difficult to resolve by compression.
In settings where kernelization does not effectively reduce the instance size, \searchname overcomes this scalability barrier.
Indeed, \searchname achieves at least a $500\%$ speedup on \textit{empty-16-16}, and a speedup by a factor of roughly $100$ on \textit{random-64-64-20} (e.g., $n=30$ and $r=1$) with acceptable suboptimality, despite the \PSPACE-completeness of \iumapf.

\subsection{Large-scale Problem of $r\ge 1$}
We evaluate the scalability of the configuration generator-based approach \searchname, and the effect of livelock resolution on large instances.
Since our \ilp-based algorithm is not scalable to a large-scale setting, we exclude \ilp from these experiments.
We generate random instances on the large maps \emph{random-64-64-20}, \emph{lak303d}, and \emph{warehouse-10-20-10-2-2} for $r\in [1,3]$, and analyzed the fraction solved within \SI{60}{\second} (\emph{success rate}) and the length of plan (\emph{makespan}).
\Cref{fig:expmain} shows the empirical result for each setting.
The main observations are as follows:
\begin{itemize}
    \item Although \algname works in sparse settings to find a solution, \algname alone rarely finds a solution in moderately dense situations.
    This suggests that \algname often gets stuck due to livelocks.
    \item \lacam-style search substantially increases the success rate in dense settings, indicating that the search scheme efficiently handles long-horizon tasks, such as livelocks.
    It also finds solutions for larger instances, including those beyond the practical reach of the CBS-based baselines.
    \item By explicitly handling livelocks, \searchname further increases the success rate. 
    This suggests that leveraging the anonymity of agents enables us to plan faster.
    However, this yields no clear benefit on maps with narrow corridors (e.g., \emph{warehouse-10-20-10-2-2} and $r{=}1$). 
    \item Related to the above point, for settings differing only in $r$ (e.g., \emph{random-64-64-20} with $n=200$), the setting with smaller $r$ exhibits a lower success rate. For small $r$, multiple agents can enter narrow corridors simultaneously, making search-stuck situations more likely.
In contrast, a large $r$ implicitly prevents such situations by enforcing distance-$r$ independence.    
\end{itemize}

\subsection{Special case $r=0$: \emph{Unlabeled} \mapf}\label{sec:exptswap}
Recall that \riumapf is equivalent to \emph{unlabeled} \mapf~when $r=0$.
Thus, we evaluate the performance of \algname against an existing configuration generator specialized for \emph{unlabeled} \mapf, namely
\textsc{TSWAP}~\cite{mapf:tswap}.\footnote{The code is from \url{https://github.com/Kei18/pytswap.git}.}
Both algorithms require an initial target assignment function, and we use the Hungarian method for this purpose.
Note that both algorithms always find solutions in these settings.
We use the time for one-step generation and the suboptimality of the plan, defined as the plan length over a trivial lower bound derived by \textsc{Bottleneck Matching}~\cite{bottleneck}, which is collision-agnostic. 
\begin{figure}[t]
    \input{draft/tikz/fig_exptswap}
    \caption{Average time per one-step configuration generation and suboptimality of TSWAP and \algname. 
    }
    \label{fig:expTSWAP}
\end{figure}

\Cref{fig:expTSWAP} shows that \algname consistently achieves a lower suboptimality than TSWAP, while the running time is substantially higher than TSWAP; as \algname allows agents to move more flexibly than TSWAP and recursively tries more candidates to move for every agent.
However, the actual running time is dominated by the one for target assignment, which takes $O(n^3)$; both algorithms take linear time per configuration generation.

\section{Conclusion}
We study \riumapf, a variant of \emph{unlabeled} \mapf that introduces an extended collision definition by distance.
Given the absence of algorithms readily available for this generalization problem, we tackle \riumapf with two approaches from the view of quality and scalability: reduction-based algorithms with compression, and configuration generator-based search.
A direction for future work is to develop a search scheme for the \emph{unlabeled} variant, efficiently handling the anonymity of agents; leave as an open question.

\section*{Acknowledgments}
This research was partially supported by JSPS KAKENHI Grant Number 25K21289 and JST PRESTO (JPMJPR2513).

\bibliography{ref-macro, AAAI/aaai2027}

\newpage
\appendix

\section*{Appendix}
\crefalias{section}{appendix}
\setcounter{secnumdepth}{2}
\section{Omitted discussions in \Cref{sec:FPT}}\label{sec:omittedFPTproofs}

\paragraph{\textsc{Token Sliding}} is a problem that asks whether one can reach a target token placement from a given initial one by moving one token to an adjacent vertex, while maintaining the placement as an independent set~\cite{reconf:hearn05}.
Formally, the problem asks whether there exists a sequence of independent sets $[S=S_0,S_1,\dots,S_\ell=T]$ such that $S_{i-1}$ and $S_i$ are adjacent for every $i\in[1,\ell]$ under the following rule.
In \textsc{Token Sliding}, for two independent sets $I_1$ and $I_2$, we say that $I_1$ and $I_2$ are adjacent if there exist vertices $u,v\in V(G)$ such that $I_1\setminus I_2=\{u\}$, $I_2\setminus I_1=\{v\}$, and $uv\in E(G)$.
This problem is similar to $1$\iumapf in the sense that we can view agents as tokens and allow at most one token to move to an adjacent vertex in each step.
Indeed, some recent works have progressed on settings motivated by \mapf~\cite{reconf:Kristan25}.

\subsection{Proofs}

We begin by proving that the reduction rules in \Cref{sec:FPT} are safe.
Here, for a graph $G = (V,E)$, we define the \emph{open neighborhood} $\nei{v}$ of $v\in V$ as $\nei{v}=\{w \in V\mid vw\in E\}$, and $\nei{V'}$ for $V'\subseteq V$ as $\neicl{V'}\setminus V'$ for the sake of simplicity.
\FPTcorrect* 

To prove~\Cref{lem: FPTcorrect}, we first present the following lemma.

\begin{lem}\label{FPT:reducebefore}
    Let $G=(V,E)$ be a graph and let $G'=(V',E')$ be the graph obtained by applying either Reduction Rule~\ref{red: adj} or Reduction Rule~\ref{red: BFS} once, which replaces a vertex subset $W$ with a single black hole $b$. If there is a plan from $S$ to $T$ for $G$, then there is a plan from $S$ to $T$ for $G'$.
\end{lem}

\begin{proof}
    Let $\Pi = [\conf{0}, \conf{1},...,\conf{\ell}]$ be a plan for $G$ from $S$ to $T$.
    Consider the sequence $\Pi' = [\conf{0}', \conf{1}',...,\conf{\ell}']$ such that:
    \begin{align*}
    \conf{k}'[i]=
    \begin{cases}
        b & \text{if}~ \conf{k}[i]\in W\\
        \conf{k}[i] & \text{otherwise.}
    \end{cases}
\end{align*} 
We now show that $\Pi'$ is a plan for $G'$ from $S$ to $T$.
First, we claim the reachability of each configuration.
For a step $k\in [0,\ell-1]$, consider an arbitrary agent $i$.
If $\conf{k}[i]\notin W$ and $\conf{k+1}[i]\notin W$ hold, then $\conf{k}'[i] = \conf{k}[i]$ and $\conf{k+1}'[i] = \conf{k+1}[i]$.
Since $\Pi$ is a plan for $G$, we have $\conf{k+1}'[i] = \conf{k+1}[i] \in \neicl{\conf{k}[i]} = \neicl{\conf{k}'[i]}$.
Otherwise, $\conf{k}[i]\in W$ or $\conf{k+1}[i]\in W$ holds.
Without loss of generality, assume that $\conf{k}[i]\in W$.
Then there are two cases: $\conf{k+1}[i]\in W$ and $\conf{k+1}[i]\in \nei{W}$.
The former case is straightforward because $\conf{k}'[i] = \conf{k+1}'[i] = b$.
In the latter case, we have $\conf{k+1}'[i] = \conf{k+1}[i]$.
Since $\conf{k}'[i] = b$ is adjacent to the vertices in $\nei{W}$ by definition, we have $\conf{k+1}'[i] = \conf{k+1}[i] \in \nei{W} =  \nei{b} \in \neicl{\conf{k}'[i]}$. 
Therefore, $\conf{k+1}'$ is a reachable configuration.
This claim holds for every $k \in [0,\ell-1]$.

Moreover, no pair of two agents $i,j$ on $P\setminus W$ is adjacent on $G$, where $P$ is a set of planets in $G$.
Since $\conf{k}[i] = \conf{k}'[i]$ for every agent $i$ with  $\conf{k}[i] \notin W$, the configuration $\conf{k}'$ is also galactic independent.
Therefore, $\Pi'$ is a plan for $G'$, completing the proof.
\end{proof}

Thus, we can transform any plan for $G$ into a plan for $G'$ in both rules.
We next show the converse, i.e., we can transform any plan for $G'$ into a plan for $G$, which completes the proof of \Cref{lem: FPTcorrect}.
We begin with the correctness of \Cref{red: adj}.
As a preparation, for a vertex $b$, we denote by $\project_t$ and $\absorb_t$ the sets of agents that leave $b$ and enter $b$ at step $t$, respectively, that is, $\project_t = \{ i \in A \mid \conf{t-1}'[i]=b \wedge \conf{t}'[i]\neq b \}$ and 
$\absorb_t = \{ i \in A \mid \conf{t}'[i]=b \wedge \conf{t-1}'[i]\neq b \}$.

\begin{lem}\label{lem:adjsafe}
    \Cref{red: adj} is safe.
\end{lem}

\begin{proof}
    Let $G$ be the initial graph and $G'$ be a graph obtained after applying \Cref{red: adj} to adjacent black holes $u,v$ of $G$.
    By \Cref{FPT:reducebefore}, if there is a plan from $S$ to $T$ for $G$, then there is a plan for $G'$.
    
    Suppose that there is a plan $\Pi' = [\conf{0}', \conf{1}',...,\conf{k}']$ for $G'$ from $\conf{0}' = S$ to $\conf{k}'$.
    Note that $\conf{k}'$ is not necessarily equal to $T$.
    For an integer $\ell$, a plan $\Pi = [\conf{0}, \conf{1},...,\conf{\ell}]$ for $G$ is said to be \emph{compatible with} $\Pi'$ if $\conf{0} = \conf{0}'$, $\conf{\ell}[i] = \conf{k}'[i]$ for each agent $i$ with $\conf{k}'[i] \neq b$, and $\conf{\ell}[i] \in \{ u, v\}$ for each agent $i$ with $\conf{k}'[i] = b$.
    By induction on $k$, we show that there is a plan $\Pi$ for $G$ compatible with $\Pi'$.
    
    The base case $k = 0$ is straightforward.
    Consider the inductive case $k > 0$.
    Let $\Phi' = [\conf{0}', \conf{1}',...,\conf{k-1}']$.
    Since $\Phi'$ is a plan from $\conf{0}'$ to $\conf{k-1}'$, there is a plan $\Phi = [\conf{0}, \conf{1},...,\conf{\ell}]$ for $G$ compatible with $\Phi'$, by the induction hypothesis.
    We extend this plan to construct a plan $\Pi$ for $G$ compatible with $\Pi'$. 
    
    Consider agent $i \in A$.
    It is clear that if $\conf{k-1}'[i]\ne b$ and $\conf{k}'[i]\ne b$ hold, then the same move can be performed in $G$, that is, define $\conf{\ell+1}[i] = \conf{k}'[i]$.
    Since $\conf{\ell}[i] = \conf{k-1}'[i] $, this move is valid. 
    Otherwise, $\conf{k-1}'[i]= b$ or $\conf{k}'[i]= b$ holds.
    If $\conf{k-1}'[i]= b$ and $\conf{k}'[i]= b$, it suffices to specify that $\conf{\ell+1}[i]=\conf{\ell}[i]$, which force $i$ to stay the former location.
    
    Now we explain the case where one of $\conf{k-1}'[i]$ and $\conf{k}'[i]$ is $b$, and the other is a neighbor of $b$.
    By symmetry, suppose that $\conf{k-1}'[i] = b$ and $\conf{k}'[i] \in \nei{b}$.
    Recall that $u$ and $v$ are contracted to $b$, and hence $\nei{b} = \nei{u} \cup \nei{v} \setminus \{u,v\}$.
    Suppose that $\conf{\ell}[i] = u$.
    If $\conf{k}'[i] \in \nei{u}$, agent $i$ can move to $\conf{k}'[i]$, that is, define $\conf{\ell+1}[i] = \conf{k}'[i]$.
    If $\conf{k}'[i] \in \nei{v}$, agent $i$ moves to $v$ first, and then move to $\conf{k}'[i]$ at the next step.
    In other words, $\conf{\ell+1}[i] = v$ and $\conf{\ell+2}[i] = \conf{k}'[i]$.
    Note that the agents not located at $u$ or $v$ at step $\ell+2$ remain at their current positions.
    The case where $\conf{\ell}[i] = v$ is analogous.
    It is obvious that $\Pi =  [\conf{0}, \conf{1},...,\conf{\ell},\conf{\ell+1},\conf{\ell+2}]$ is compatible with $\Pi'$.
    Moreover, since the agents on planets $P$ at steps $\ell+1$ and $\ell+2$ follow $\conf{k}'[i]$, we conclude that $\conf{\ell+1}$ and $\conf{\ell+2}$ are galactic independent.
\end{proof}

Next, we discuss the correctness of \Cref{red: BFS}.
To this end, we give two lemmas that impose a constraint on plans in $G'$.

\begin{lem}\label{lem:nextempty}
    Let $\Pi = [\conf{0},\conf{1},...,\conf{\ell}]$ be a plan from $S$ to $T$ for $G'$. 
    Then, $\Pi$ can be transformed to a plan $\Pi'= [\conf{0}',\conf{1}',...,\conf{\ell'}']$ for $G$ such that no two consecutive configurations have an agent in the neighborhood of $b$.
\end{lem}
\begin{proof}
    Consider a step $t$ such that $\conf{t}\cap \nei{b}\ne \emptyset$.
    We now give five intermediate configurations $\mathcal{Q}_{t}^1,\mathcal{Q}_{t}^2,\mathcal{Q}_{t}^3,\mathcal{Q}_{t}^4,\mathcal{Q}_{t}^5$ between $\conf{t}$ and $\conf{t+1}$:
    \begin{align*}
    &\mathcal{Q}_{t}^1[i]=
    \begin{cases}
        b & \text{ if } \conf{t}[i]\in \neicl{b}\\
        \conf{t}[i] & \text{ otherwise }
    \end{cases}\\
    &\mathcal{Q}_{t}^2[i]=
    \begin{cases}
        b & \text{ if } \conf{t}[i]\in \neicl{b}\\
        \conf{t+1}[i] & \text{ otherwise }
    \end{cases}\\
    &\mathcal{Q}_{t}^3[i]=
    \begin{cases}
        b & \text{ if } \conf{t}[i]\in \neicl{b} ~\lor~ \conf{t+1}[i]\in \neicl{b}\\
        \conf{t+1}[i] & \text{ otherwise }
    \end{cases}\\
    &\mathcal{Q}_{t}^4[i]=
    \begin{cases}
        b & \text{ if } \conf{t+1}[i]\in \neicl{b}\\
        \conf{t}[i] & \text{ otherwise }
    \end{cases}\\
    &\mathcal{Q}_{t}^5[i]=
    \begin{cases}
        b & \text{ if } \conf{t+1}[i]\in \neicl{b}\\
        \conf{t+1}[i] & \text{ otherwise }
    \end{cases}
\end{align*}
    Note that the galactic independence of these five configurations is followed by the galactic independence of $\conf{t}$ and $\conf{t+1}$.
    We claim that the sequence $[\conf{t},\mathcal{Q}_{t}^1,\mathcal{Q}_{t}^2,\mathcal{Q}_{t}^3,\mathcal{Q}_{t}^4,\mathcal{Q}_{t}^5, \conf{t+1}]$ is a plan for $G$ and satisfies the required conditions, that is, no two consecutive configurations have an agent in the neighborhood of $b$.
    For this purpose, we classify the agents into the following four types:
    \begin{enumerate}
    \item $\conf{t}[i]\in \neicl{b} ~\land~ \conf{t+1}[i]\in \neicl{b}$,
    \item $\conf{t}[i]\in \neicl{b} ~\land~ \conf{t+1}[i]\notin \neicl{b}$,
    \item $\conf{t}[i]\notin \neicl{b} ~\land~ \conf{t+1}[i]\in \neicl{b}$, and
    \item $\conf{t}[i]\notin \neicl{b} ~\land~ \conf{t+1}[i]\notin \neicl{b}$,
\end{enumerate}
and verify their locations in each configuration.
For each agent $i$, a sequence $[\conf{t}[i],\mathcal{Q}_{t}^1[i],\mathcal{Q}_{t}^2[i],\mathcal{Q}_{t}^3[i],\mathcal{Q}_{t}^4[i],\mathcal{Q}_{t}^5[i],\\\conf{t+1}[i]]$ of vertices consists of:\\
1. $[\conf{t}[i], b,b,b,b,b,\conf{t+1}[i]]$, \\
2. $[\conf{t}[i], b,b,b,\conf{t}[i],\conf{t+1}[i],\conf{t+1}[i]]$, \\
3. $[\conf{t}[i], \conf{t}[i],\conf{t+1}[i],b,b,b,\conf{t+1}[i]]$, and\\
4. $[\conf{t}[i], \conf{t+1}[i],\conf{t+1}[i],\conf{t+1}[i],\conf{t}[i],\conf{t+1}[i],\conf{t+1}[i]]$.\\
Thus, every agent either moves to an adjacent vertex along an edge or stays at its current vertex, ensuring the reachability.
Moreover, if agent $i$ is located in a neighbor of $b$, it moves out of the open neighborhood of $b$ in the next step.
The lemma follows by applying this procedure to every $t \in 
[0, \ell-1]$ such that $Q_t \cap N(b) \neq \emptyset$ holds.
\end{proof}

In what follows, we transform the plan on $G'$ obtained above into one that satisfies additional restrictions on $\project_t$ and $\absorb_t$.

\begin{lem}\label{lem: blackhole}
Suppose that $G'$ is obtained by applying \Cref{red: BFS} to replace a component $C$ with a black hole $b$.
Let $\Pi=[\conf{0},\conf{1},...,\conf{\ell}]$ be a plan for $G'$ from $S$ to $T$.
One can transform it to a plan $\Pi'=[\conf{0}',\conf{1}',...,\conf{\ell'}']$ with the following two conditions:
\begin{enumerate}
        \item no two consecutive configurations of $\Pi'$ have an agent in the neighborhood of $b$, and 
        \item $|\project_t|+|\absorb_t|\le 1$ for every $t \in [0,\ell']$.
    \end{enumerate}
\end{lem}
\begin{proof}
    Assume that $\Pi$ satisfies condition 1 by applying \Cref{lem:nextempty}.
    Consider a step $t\in [1,\ell]$ such that $|\project_t|+|\absorb_t|> 1$.
    Note that it holds that $t\notin \{1,\ell\}$ by the fact that $b\notin \neicl{S\cup T}$, and $|\project_t| = 0 \lor |\absorb_t| = 0$ by condition 1.
    Here, by symmetry, we assume that $|\absorb_t|>0$ and $|\project_t|=0$.
    Let $i_1\in \absorb_t$ be an agent, and consider the sequence $[\conf{t-1},\mathcal{Q}_{t-1}^6,\mathcal{Q}_{t-1}^7, \conf{t}]$ with the following:
    \begin{align*}
     &\mathcal{Q}^6_{t-1}[i]=
    \begin{cases}
        b & \text{ if } i\in \absorb_t\setminus \{i_1\}\\
        \conf{t-2}[i] & \text{ otherwise }
    \end{cases}\\
    &\mathcal{Q}^7_{t-1}[i]=
    \begin{cases}
        b & \text{ if } i\in \absorb_t\setminus \{i_1\}\\
        \conf{t-1}[i] & \text{ otherwise. }
    \end{cases}
\end{align*}
In other words, each agent $i \notin \absorb_t\setminus \{i_1\}$ follows the sequence of moves $\conf{t-1}[i] \to \conf{t-2}[i] \to \conf{t-1}[i] \to \conf{t}[i]$, and $i\in \absorb_t\setminus \{i_1\}$ moves from $\conf{t-1}[i]$ to $b=\conf{t}[i]$; hence, this sequence of configurations is a plan for $G'$.
Moreover, the plan satisfies the condition 1, since $\conf{t-2}\cap \nei{b}=\emptyset$ holds from \Cref{lem:nextempty}.

By inserting this subsequence between $\conf{t-1}$ and $\conf{t}$, the size of $\absorb_t$ decreases by one.
Applying this process iteratively, we eventually achieve $|\absorb_t| = 1$.
We thereby obtain a sequence in which the number of steps satisfying $|\project_t|+|\absorb_t|>1$ is strictly smaller than that of $\Pi$.
By symmetry, a similar argument is applicable to the case where $|\absorb_t|=0$ and $|\project_t|>0$.
Therefore, we can obtain a plan for $G'$ that satisfies conditions~1 and~2.
This completes the proof.
\end{proof}

We now demonstrate that \Cref{red: BFS} is safe using the aforementioned lemmas.

\begin{lem}\label{lem:BFSsafe}
    \Cref{red: BFS} is safe.
\end{lem}
\begin{proof}
    Let $G_1$ be a graph, and let $G_2$ be a graph after \Cref{red: BFS} is applied to a component $C$.
    By \Cref{FPT:reducebefore}, if there is a plan from $S$ to $T$ for $G$, then there is a plan for $G'$.
    
    Assume that there is a plan $\Pi = [\conf{0},\conf{1},...,\conf{\ell}]$  for $G_2$ from $S$ to $T$.
    By \Cref{lem: blackhole}, we can assume that $\Pi$ satisfies the following two conditions:
    \begin{inparaenum}[(i)]
        \item no two consecutive configurations of $\Pi$ have an agent in the neighborhood of $b$, and 
        \item $|\project_t|+|\absorb_t|\le 1$ for every $t \in [0, \ell]$.
    \end{inparaenum}
    Consider a vertex $v \in L_k\cap C$ with $k\ge 2n+3$, and a shortest path $R$ length from $v$ to a vertex $u\in S\cup T$.
    Note that for every pair of distance-$1$ configurations $X$ and $Y$ in $R$ with $|X|=|Y|\le n$, there is a plan from $X$ to $Y$ that uses only the vertices of $R$.
    Accordingly, for an agent $i$ that moves to $b$ at step $t$ on $G_2$, we will instead route $i$ to a vertex on the path $R$ on $G_1$.

    Here, for an agent $i$ on $G_2$ and step $t'\in [0,\ell-1]$, if $\conf{t'}[i] \neq b$ and $\conf{t'+1}[i]\ne b$, then the same move can be done on $G_1$, since $G[V(G_1)\setminus C]$ is isomorphic to $G[V(G_2)\setminus \{b\}]$.
    Now we assume that $\conf{t'}[i]=b$ or $ \conf{t'+1}[i]= b$.
    Since $b\notin S\cup T$, there is a step $t \le t'+1$ such that $\conf{t}[i]\in \nei{b}$ and $\conf{t+1}[i]=b$.
    We now show that the agent $i$ on $G_1$ can \emph{imitate} such a move on $C$, that is, agent $i$ can move to some vertex of $R$ while ensuring that every other agent $j$ currently contained in $C$ is also located at some vertex in $R$.

    Consider shortest paths from $\conf{t}[i]$ to every vertex of $R$, and let $R'$ be the shortest one among them.
    We move agent $i$ along this path $R'$ and place it at a vertex $r'$ on the path $R$.
    Once $i$ has moved to some vertex of $C$, for each agent $j\neq i$ contained in $\neicl{C}\setminus C$, we move it to $\conf{t+1}[j]$. 
    This ensures that $i$ and $j$ are never on adjacent vertices, since $\conf{t+1}[j]\notin \neicl{b}$ on $G_2$, and thus $\conf{t+1}'[j]\notin \neicl{C}$ on $G_1$.
    Before $i$ reaches the endpoint $r'$, agent $i$ could become adjacent to an agent $k$ on $R$, violating the independence.
    In this case, $i$ is at the vertex $v'$ on $R'$ adjacent to $r'$.
    Note that $N(v')\cap R$ consists of at most three consecutive vertices on $R$; otherwise, we would obtain a path shorter than $R$ via $v'$.
    Accordingly, we ensure that no agent is placed on $N(v')\cap R$ before agent $i$ moves to $v'$.
    Here, after removing at most three vertices in $N(v')\cap R$ from $R$, there remain at least $2n-2$ vertices, which can accommodate an independent set of size $n-1$.
    Then, we move the agents on $R$ so that they form this independent set.
    After completing these moves, agent $i$ can move from $v'$ to $r'$.
    Therefore, the agent $i$ on $G_1$ can imitate a move on $G_2$ from a vertex in $\nei{b}$ to $b$.
    
    Until agent $i$ moves from $b$ to a vertex in $\nei{b}$ on $G_2$, it can be kept at some vertex of $R$.
    This holds because, whenever another agent $j$ moves to $b$, agent $j$ can imitate its move on $G_2$.

    From here on, we consider the reverse move, that is,  $i$ moves from $b$ to a vertex in $v\in \nei{b}$ at step $t_1$.
    we try to move $i$ from a vertex on $R$ to $v\in \nei{b}$.
    However, since there may be other agents on $R$, it might be impossible to move $i$ without violating the independence constraint. 
    Nevertheless, by considering the shortest path $R'$ from $v$ to a vertex $R$ again, we can see that at least one agent $j$ can be moved to $v$.
    Then, consider the sequence obtained by swapping $\conf{t'}[i]$ and $\conf{t'}[j]$ in every configuration $\conf{t'}$ with $t'\ge t_1$.
    This sequence is also a plan from $S$ to $T$ on $G_2$, and it yields a sequence in which $j$ moves to $v$ on $G_1$.
    Here, every agent $i$ holds that $\conf{\ell}[i]\notin \neicl{b}$, since $b\notin \neicl{T}$.
    Eventually, every agent moves from $b$ to some planet.
    At some step, agent $i$ can imitate the move of some agent moving from $b$ to $\nei{b}$.
    This completes the proof.
\end{proof} 

\Cref{lem:adjsafe} and \Cref{lem:BFSsafe} immediately lead to \Cref{lem: FPTcorrect}.

We bound the running time to obtain an instance after applying \Cref{red: adj,red: BFS} until no updates are possible.
We first compute the layers in time $O(|V(G)|+|E(G)|)$ using Breadth-first search.
\Cref{red: adj} searches for adjacent black holds and \Cref{red: BFS} contract a component to a vertex.
Thus, each application of the rules requires $O(|V(G)|+|E(G)|)$ time.
The running time is $O(|V(G)|+|E(G)|)$. 

We now provide an upper bound on the size of the obtained instance. 
We first claim that the size of $B$ is bounded by $\Delta\cdot |P|$.
Since there is no edge $b_1b_2\in B$ such that $b_1,b_2$ are both in $B$, the remaining black hole $b\in B$ holds that $\nei{b}\subseteq P$.
Moreover, we can see that the degree of vertex $v\in P$ does not increase by applying \Cref{red: adj,red: BFS}.
Thus, $|B|$ can be upper bounded by $|\neicl{P}|\le \Delta |P|$.
It remains only to bound $|P|$.

\kerneldeg*
\begin{proof}
    Consider the instance obtained after applying the rules until no updates are possible.
    Since every planet $p\in P$ is contained in the layer $L_k$ with $k \le 2n+2$, the size of $P$ is bounded from above by $\sum_{0\le k\le 2n+2} |L_k|$.
    Note that, since the maximum degree of $G[P]$ is at most $\Delta$, we have $|L_k| \le \Delta \cdot |L_{k-1}|$ for each $k \in [1,2n+2]$.
    It follows from $L_0 = S \cup T$ that 
    \begin{align*}
        & \sum_{0\le k\le 2n+2}|L_k|\\
        \le & \sum_{0\le k\le 2n+2} |S\cup T| \cdot \Delta^k \\
        \le &~ 2n \cdot \frac{\Delta^{2n+3}-1}{\Delta-1}\\
        \le &~ 2n \cdot \Delta^{2n+3}.
    \end{align*}
    This completes the proof.
\end{proof}

Here, when an induced grid graph is given, the above bound is overestimated. 
Any vertex within distance $k$ from a vertex $v$ lies inside the grid of width $2k$ centered at $v$.
Thus, $\sum_{0\le k\le 2n+2}|L_k| = O(|S\cup T|\cdot n^2) = O(n^3)$.
This leads to the following theorem.

\kernelgrid*

An {\FPT} algorithm for $1\iumapf$ immediately follows.
\galFPT*
\begin{proof}
    Let $(G,S,T)$ be an instance of ${1\iumapf}$.
    We transform it to an instance $(G', S, T)$ of {\giumapf}, where a set of planets is $V(G)$ and a set of black holes is $\emptyset$.
    By applying \Cref{thm: kernel}, we obtain a kernel with $\Delta^{O(n)}n$ vertices.
    Then the claim follows from a brute-force search.
\end{proof}

\paragraph{Remark.}
We note that the lemma established above does not specialize to \giumapf and also applies to \textsc{Galactic Token Sliding}, a \emph{galactic} variant of \textsc{Token Sliding}. 
Therefore, we can say that \textsc{Galactic Token Sliding} admits a kernel with $\Delta^{O(n)}n$ vertices in general, and with $O(n^3)$ vertices when an induced grid graph is given.
This claim improves the known upper bound on the kernel size $\Delta^{O(n^2)}n$ in \cite{gal:Bartier23}.
\section{Omitted pseudocode in \Cref{sec:PIBT}}\label{sec:omittedcode}
In order, \Cref{alg:deadlock,alg:TSWAP,alg:lacam,alg:PIBTlacam} illustrate pseudocodes for deadlock detection and resolution, detection of $\textsc{SWAP}(i,v)$, the strategy of \searchname, and a modified configuration generator using \algname.
All implementations are shown in the code appendix.

\begin{algorithm}[H]
		\caption{detect a deadlock starting with $i$}
		\label{alg:deadlock}
		\begin{algorithmic}[1]  
			\Require agent $i\in [n]$, vertex $u\in \neicl{\qfrom{i}}$
			\Ensure list of agents $L$, that there is a deadlock, or $\bot$ (means empty list)
            \State $v\leftarrow \nextvr{r}(u,\goal{i})$\Comment{$\dist(u,v)=r$}
            \If{$\exists j$ s.t. $\qfrom{j}=v\land j\ne i\land \qto{j}=\bot$} $L\leftarrow [j]$
            \Statex\Comment{$\exists j \Rightarrow\dist(\qfrom{i},u)=r+1$}
            \Else~ \Return $\bot$
            \EndIf
            \While{$j\ne i$}
            \State $v'\leftarrow \nextvr{r+1}(\qfrom{j},\goal{j})$
            \If{$\neg(\exists j$ s.t. $\qfrom{j}=v'\land \qto{j}=\bot)$}~\Return $\bot$
            \ElsIf{$\exists a$ s.t. $L[a]=j$}~\Return $\bot$
            \Else ~$L\leftarrow L+[j]$
            \EndIf
            \EndWhile
            \State \Return $L$
		\end{algorithmic}
	\end{algorithm}

    \begin{algorithm}[H]
    {\small
		\caption{rule 2: swap detection}
		\label{alg:TSWAP}
		\begin{algorithmic}[1]  
			\Require agent $i$ and vertex $v\in \neicl{\qfrom{i}}$ 
			\Ensure an agent $j\in A$ that needs to swap the target
            \State $u\leftarrow \nextvr{r}(v,\goal{i})$\Comment{$\dist(u,v)=r$}
            \If{$\exists j$ s.t. $\qfrom{j}=u\land \qto{j}= \bot$}
            \If{$\goal{j}=u$}~\Return $j$ \Comment{$\exists j \Rightarrow\dist(\qfrom{i},u)>r$}
            \EndIf
            \EndIf
            \State \Return $\bot$
		\end{algorithmic}
        }
	\end{algorithm}

{
\begin{algorithm}[ht]
\caption{LaCAM for \riumapf}
\label{alg:lacam}
\begin{algorithmic}[1]
\small
\State initialize \open, \explored
\State $\N\init \leftarrow \langle
S, \llbracket~\mathcal{C}\init~\rrbracket , g\init, \mathcal{B}\init
\rangle$ \Comment{$\mathcal{C}\init, \mathcal{B}\init$: no constraints}
\label{algo:lacam:init-node}
\Statex $\langle
\text{config, constraints, assignment, banned\_list}
\rangle$
\State $\text{OPEN}.\push(\N\init)$;~~$\explored[S] = \N\init$
\While{$\text{OPEN} \neq \emptyset$}
\State $\N \leftarrow \text{OPEN}.\funcname{top}()$
\IfSingle{$\N.\config = T$}{\Return plan}
\label{algo:lacam:backtrack}
\IfSingle{$\N.\tree = \emptyset$}{$\text{OPEN}.\pop()$;~\Continue}
\State $\C \leftarrow \N.\tree.\pop()$
\State $\funcname{update\_constraints}(\N, \mathcal{C})$
\label{algo:lacam:lowlevel-search}
\Comment{constraints synthesis}
\State $\Q\new, g\new \leftarrow \funcname{configuration\_generator}(\N, \mathcal{C})$
\label{algo:lacam:config-generator}
\IfSingle{$\Q\new = \bot$}{\Continue}
\State \textcolor{blue}{$\funcname{detect\_livelock\_and\_reassignment}(\Q\new, g\new)$}\label{lacam:livelock}
\Statex\Comment{for resolving livelock}
\IfSingle{$\explored[\Q\new] \neq \bot$}{\Continue}
\State $\N\new \leftarrow \langle
\Q\new,~
\llbracket~\mathcal{C}\init~\rrbracket,~
g\new,~
\mathcal{B}\init
\rangle$
\State $\text{OPEN}.\push(\N\new)$;\;$\explored[\Q\new] = \N\new$
\EndWhile
\State \Return \texttt{NO}\_\texttt{PLAN}
\label{algo:lacam:unsolvable}

\Function{\funcname{configuration\_generator}}{$\N, \mathcal{C}$}
\State Define $\Q\suf{to}[i]\coloneqq v \text{ if } (i,v) \in \mathcal{C} \text{ else } \bot$
\Statex \Comment{$\mathcal{C}$: list of $(\text{agent},\text{destination} )$}
\State $\Q, g \leftarrow \N.\text{config}, \N.g$
\State $\Q\new, g\new\leftarrow\Call{\algname}{\Q, T,\N.g}$ initialzed with $\Q\suf{to}$
\State\Return $\Q\new, g\new$
\EndFunction
\end{algorithmic}
\end{algorithm}
}

\begin{algorithm}[ht]
    \caption{Configuration generator for \lacam}
    \label{alg:PIBTlacam}
    {\small
    \begin{algorithmic}[1]  
        \Require configuration $\mathcal{Q}^{\text{from}}$, constraints $\mathcal{C}$, goals $T$, temporal target assignment $g:[1,n]\to T$, order $P$ of agents
        \Ensure configuration $\mathcal{Q}^{\text{to}}$, new assignment $g'$ \\\Comment{initialized by $((\bot)^n, g, p)$}
        \For{$(i,v)\in \mathcal{C}$} $\qto{i}=v$
        \EndFor
        \If{$\exists i,j\in A$ s.t. $i\ne j\land \dist(\qto{i}, \qto{j})\le r$}
        \State\Return $\bot, g$ \Comment{not a distance-$r$ independent config}
        \EndIf
        \For{$i\in P$} \Comment{respect the order $P$, indicates the priority}
        \If{$\qto{i}=\bot$}~$\Call{\procname}{i,[~],g}$        
        \EndIf
        \EndFor
        \If{$\exists i,j\in A$ s.t. $i\ne j\land \dist(\qto{i}, \qto{j})\le r$}
        \State\Return $\bot, g$ \Comment{not a distance-$r$ independent config}
        \EndIf
        \State \Return $\mathcal{Q}^{\text{to}},g$
    \end{algorithmic}
    }
\end{algorithm}
\section{Omitted Discussions in \Cref{sec:PIBT}}\label{sec:omittedPIBTproofs}

For a timestep $\tau$, let $\conf{\tau}$ and $g_{\tau}$ denote the inputs $\mathcal{Q}^{from}$ and $g$ at $\tau$. 
Moreover, we use the inverse function $g^{-1}_{\tau}$ of $g_{\tau}$.
For a vertex subset $S$ of a graph $G$, let $\indG{S}$ denote the subgraph induced by $S$.
We also use a notation used in \Cref{sec:omittedFPTproofs}; for a graph $G = (V,E)$, we define the \emph{open neighborhood} $\nei{v}$ of $v\in V$ as $\nei{v}=\{w \in V\mid vw\in E\}$, and $\nei{V'}$ for $V'\subseteq V$ as $\neicl{V'}\setminus V'$ for the sake of simplicity.
\Cref{lem: correctconf} provides the soundness of \Cref{alg:PIBT}. 

Note that, since $\qto{i}$ is determined sequentially for agent $i$, we can construct a total order based on the times at which the procedure $\Call{\procname}{i,\ast, \ast}$ for agent $i\in A$ returned $\valid$ or $\invalid$; without loss of generality, we now assume that agent $i \in A $ is the $i$-th agent whose destination $\qto{i}$ is determined.\footnote{Note that this algorithm is designed not to be parallelized. Thus, for every pair of agents $x,y\in [n]$, $\qto{x}$ and $\qto{y}$ are not determined simultaneously.}
Furthermore, we can see that once $\procname(i,*,*)$ returns a value in $\{\valid, \invalid\}$ for agent $i\in A$, it will not be called again for agent $i$.

We start with the following simple observations, which are easily derived by the if-statements in lines \ref{pibt:dist},\ref{pibt:rotation},\ref{pibt:recursive}--\ref{pibt:valid}.
\begin{obsv}\label{obs:pibtif}
    If $\Call{\procname}{i, S,\ast}$ returns $\valid$ for an agent $i\in A$ and a list $S$ of agents, then every agent $j\in A\setminus \{i\}$ holds neither $\qto{j}\in \neiclr{r}{\qto{i}}$ nor $\qfrom{j}\in \neiclr{r}{\qto{i}}~\land~j\in S$.
\end{obsv}
\begin{obsv}\label{obs:pibtorder}
Consider a call $\Call{\procname}{i,S,\ast}$ for an agent $i\in A$.
For any list $S'$ of agents and an agent $j\in A$, the two statements are equivalent:
\begin{enumerate}
    \item $S'+[j]$ is a prefix of $S$, and 
    \item $\Call{\procname}{j,S',\ast}$ is called before $\Call{\procname}{i,S,\ast}$ and $j>i$.
\end{enumerate}
\end{obsv}
\subsection{Omitted Proofs in \cref{sec:proof}}

\correctconf*
\begin{proof}
Reachability follows immediately since each agent $i$ is assigned either
a vertex in $N[\qfrom{i}]$ or $\qfrom{i}$ itself.
Regarding the distance-$r$ independence, 
let $Q_k$ be a configuration defined as $Q_k = \{\qto{i}\mid 1\le i\le k\}\cup \{\qfrom{i}\mid k+1\le i\le n\}$ for $k\in [0,n]$.
Now we show the following claim ($\spadesuit$): $Q_k$ is distance-$r$ independent for every $k\in [0,n]$.
Note that $Q_k$ represents the hypothetical configuration where every agent $i$ without a determined destination (i.e., $\qto{i}$ is temporarily assigned, or $\qto{i}=\bot$) remains at their current locations $\qfrom{i}$.
We prove the claim by induction on $k$.
The base case $k=0$ is trivial since $Q_0=\mathcal{Q}^{\rm from}$ holds by the initial assumption.

We now show that the claim ($\spadesuit$) holds for $k > 0$, assuming that $Q_{k-1}$ is distance-$r$ independent.
Consider the case when $\Call{\procname}{k,S,g}=\invalid$.
In this case, \Cref{alg:PIBT} updates $\qto{k}=\qfrom{k}$, and thus $Q_{k}=Q_{k-1}$.

We next move to the case $\Call{\procname}{k,S,g}=\valid$: agent $k$ is assigned to a vertex $v\in \neicl{\qfrom{k}}$.
Since $\Call{\procname}{k,S,g}$ returns $\valid$, there is no agent $i\in A\setminus \{k\}$ such that $\dist(\qto{k}, \qto{i})\le r$ or $\qfrom{i} \in \neiclr{r}{v} \land i\in S$ from \Cref{obs:pibtif}.
Here, every index $i\in[k+1,n]$ falls into one of two cases: either $\Call{\procname}{i,\ast,\ast}$ has already been called, or it has not.
Let $D$ be the set of agents that belong to the former.
First, consider an agent $i\in[k+1,n]\setminus D$.
Since the call for $i$ has not occurred, $\qto{i}=\bot$.
Assume for a contradiction that $\qfrom{i}\in \neiclr{r}{v}$, implying the if-statement in line~\ref{pibt:existagent} is true in the for-loop (line~\ref{pibt:checknei}).
Then, line~\ref{pibt:recursive} calls
$\Call{\procname}{i,S+[k],\ast}$, implying $i<k$, a contradiction.
Hence, $\qfrom{i}\notin\neiclr{r}{v}$.

It remains to consider the agents in $D$.
Now we show that $D\subseteq S$.
Consider an agent $i\in D$, and let $S'$ be the second argument of the call $\Call{\procname}{i, \ast, \ast}$.
Since $i\in D$, this call occurs before $\Call{\procname}{k,S,\ast}$, while $i>k$.
By Observation~\ref{obs:pibtorder}, $S'+[i]$ is a prefix of $S$, and hence $i\in S$.
Therefore, $D\subseteq S$, and \Cref{obs:pibtif} implies
$\qfrom{i}\notin\neiclr{r}{v}$ for every $i\in D$.

Therefore, for every $i\in[k+1,n]$, we have
$\qfrom{i}\notin\neiclr{r}{v}$.
Together with \Cref{obs:pibtif} and the induction hypothesis,
this shows that $Q_k$ is distance-$r$ independent.
\end{proof}

For a running time, we show the following.\footnote{Note that the theoretical upper bound of $\alpha$ is $O(n^3)$; however, this is too conservative in practice, as we can see in the experiment.}
\runtime*
\begin{proof}
    Within one step, the function {\procname} is called exactly $n$ times, either recursively or by the top-level procedure.
    We analyze the running time of each call, excluding the time spent in its recursive calls.
    We start by analyzing the time for each candidate $v\in \neicl{\qfrom{i}}$ for each agent $i$.
    It is sufficient that we first compute $\neiclr{r}{v}$, and for every vertex $u \in \neiclr{r}{v}$, store the agents that satisfy either $\qfrom{j}=u$ or $\qto{j}=u$.
    This can be done in at most $O(\Delta_r)$ time for each, since $|\neiclr{r}{v}|\le \Delta_r$.
    Moreover, the swap detection in line~\ref{pibt:swap} needs $O((r+1)\Delta)$ time, since the function $\nextv(u,v)$ returns a vertex in $O(\Delta)$ time among at most $\Delta$ neighbors.
    Thus, each iteration of the for-loop in line~\ref{pibt:loop_candidate} takes $O(\Delta_r+(r+1)\Delta)=O(\Delta_r(r+1))$ time.
    Since $|N[\qfrom{i}]|\leq \Delta+1=O(\Delta)$, the loop considers $O(\Delta)$ candidates.
    Thus, the running time of each call excluding its recursive calls is $O(\Delta_r\Delta(r+1))$.

    Here, we associate the time required for each operation with the agent processed by the call in which that operation is called.
    For a call $\Call{\procname}{i,S,g}$, we charge agent $i$ for the time spent on the operations performed directly in this call, excluding the time spent in its recursive calls.
    Since each operation is performed in at most one such call, the total running time is bounded by the sum of these operation times over all agents.
    Hence, the running time of recursive procedure is bounded by $O(\Delta_r\Delta(r+1))\cdot n = O(\Delta_r\Delta(r+1)n)$.
    Including the deadlock resolution, the total running time is $O(\Delta_r\Delta(r+1)n+\alpha)$.
\end{proof}

\subsection{Omitted Proof of {\cref{sec:lacam}}}
\lacamcomplete*
\begin{proof}
We show that IU-LaCAM generates all configurations reachable
from every configuration $Q$ in finite time.
Since the number of configurations is finite, this implies that
\searchname eventually finds a plan from $S$ to $T$ whenever there is a plan.

We proceed to prove the initial claim.
To this end, we prove the following:
\begin{inparaenum}
    \item[\emph{(i)}] the number of node reinserts (in line \ref{livelock:reset}) is bounded, and
    \item[\emph{(ii)}] the number of possible constraints are bounded.
\end{inparaenum}
The completeness immediately follows since every neighboring configuration can be generated by the constraint that specifies $\qto{i}$ for every agent $i\in A$.

For \emph{(i)}, we show that $\sum_{i\in A}|\mathcal{B}[i]|$ is strictly increasing when a node is reinserted in line \ref{livelock:reset} of \Cref{algo:livelock}.
It is clear that $D$ is not empty; otherwise $\Q=T$, thus a plan from $S$ to $T$ is already obtained.
Therefore, there is an agent $i$ such that $|B[i]|$ increases, and hence $\sum_{i\in A}|B[i]|$ increases.
Combining with the fact that $g=g^{\mathrm{ans}}$, we have $\sum_{i\in A}|B'[i]|>
\sum_{i\in A}|B^{\mathrm{ans}}[i]|$.
The upper bound of $\sum_{i\in A}|\mathcal{B}[i]|$ is $|A|^2$, thus the number of node reinserts is also bounded by $|A|^2$.

Next, for \emph{(ii)}, constraints have information that includes a subset of agents and their next locations.
Here, the number of subsets of agents is bounded by $2^{|A|}$, and the number of candidate next locations is bounded by $\Delta+1$ for each agent.
Thus, the number of constraints is bounded by $2^{|A|}\cdot (\Delta+1)^{|A|}$.

Thus, every neighboring configuration is generated in finite time,
which proves the initial claim and completes the proof.
\end{proof}

\end{document}